\documentclass[11pt]{article}
\usepackage{amsmath,amsthm,amsfonts,amssymb}
\usepackage{fullpage}
\usepackage{liyang}
\usepackage{framed}
\usepackage{verbatim}

\usepackage[usenames,dvipsnames]{xcolor}
\makeatletter
\newtheorem*{rep@theorem}{\rep@title}
\newcommand{\newreptheorem}[2]{
\newenvironment{rep#1}[1]{
 \def\rep@title{#2 \ref{##1}}
 \begin{rep@theorem}\itshape}
 {\end{rep@theorem}}}
\makeatother

\theoremstyle{plain}



\def\colorful{0}

\ifnum\colorful=1

\newcommand{\orange}[1]{{\color{orange}{#1}}}
\newcommand{\blue}[1]{{{\color{blue}#1}}}
\newcommand{\red}[1]{{\color{red} {#1}}}
\newcommand{\green}[1]{{\color{green} {#1}}}

\fi
\ifnum\colorful=0

\newcommand{\orange}[1]{{{#1}}}
\newcommand{\blue}[1]{{{#1}}}
\newcommand{\red}[1]{{{#1}}}
\newcommand{\green}[1]{{{#1}}}
\fi

\usepackage{boxedminipage}

\newcommand{\ignore}[1]{}
\newcommand{\lnote}[1]{\footnote{{\bf \color{blue}Li-Yang}: {#1}}}
\newcommand{\rnote}[1]{\footnote{{\bf \color{orange}Rocco:} {#1}}}

\newreptheorem{theorem}{Theorem}
\newtheorem*{theorem*}{Theorem}
\newreptheorem{lemma}{Lemma}
\newreptheorem{proposition}{Proposition}
\newtheorem*{noclaim*}{Claim}

\renewcommand{\span}{\mathrm{span}}
\newcommand{\duo}{d_{\mathrm{UO}}}
\def\wt{\mathrm{wt}}
\def\dham{d_{\mathrm{Ham}}}

\def\mp{{{h}}}
\def\bQ{\mathbf{Q}}

\begin{document}

\title{Boolean function monotonicity testing requires (almost)\\ $n^{1/2}$ non-adaptive queries}

\author{
Xi Chen\thanks{Supported by NSF grant CCF-1149257 and a Sloan research fellowship. }\\
Columbia University\\
{\tt xichen@cs.columbia.edu}\\
\and
Anindya De\thanks{ Supported by NSF  under agreements Princeton University Prime Award No. CCF-0832797 and Sub-contract
No. 00001583. }\\
IAS\\
{\tt anindya@math.ias.edu}\\
\and
Rocco A.\ Servedio\thanks{Supported by NSF grants CCF-1115703 and CCF-1319788.}\\
Columbia University\\
{\tt rocco@cs.columbia.edu}\\
\and
Li-Yang Tan\thanks{This work was done while the author was at Columbia University, supported by NSF grants CCF-1115703 and CCF-1319788.}\\ Simons Institute, UC Berkeley \\
 {\tt liyang@cs.columbia.edu}
 }

\maketitle

\begin{abstract}
We prove a lower bound of $\Omega(n^{1/2 - c})$, for all $c>0$, on the query complexity of (two-sided error)
non-adaptive algorithms for testing whether an $n$-variable Boolean function is monotone versus constant-far
from monotone.  This improves a $\tilde{\Omega}(n^{1/5})$ lower bound for the same problem that was recently given in \cite{CST14}
and is very close to $\Omega(n^{1/2})$, which we conjecture is the optimal lower bound for this model.
\end{abstract}

\section{Introduction}

\subsection{Motivation and background}
Monotonicity testing of Boolean functions $f: \{-1,1\}^n \to \{-1,1\}$ is one of the most natural and well-studied problems in Property Testing. Introduced by Goldreich, Goldwasser, Lehman, and Ron in 1998~\cite{GGL+98}, this problem is concerned with the query complexity of determining whether a Boolean function $f$ is \emph{monotone} or \emph{far from monotone}. Recall that $f$ is monotone if $f(X) \le f(Y)$ for all $X \prec Y$, where $\prec$ denotes the bitwise partial order on the hypercube. We say that $f$ is $\eps$-close to monotone if $\Pr[f(\bX)\ne g(\bX)] \le \eps$ for some monotone Boolean function $g$, where the probability is over a uniform draw of $\bX$ from $\{-1,1\}^n$, and that $f$ is $\eps$-far from monotone otherwise. We are interested in query-efficient randomized algorithms for the following task:

\begin{quote}
{\it Given as input a distance parameter $\eps > 0$ and oracle access to an unknown Boolean function $f\isafunc$, output {\sf Yes} with probability at least $2/3$ if $f$ is monotone, and ${\sf No}$ with probability at least $2/3$ if $f$ is $\eps$-far from monotone.}
\end{quote}

The work of Goldreich et al.~\cite{GGL+98} proposed a simple ``edge tester'' for this task and proved an $O(n^2\log(1/\eps)/\eps)$ upper bound on its query complexity, subsequently improved to $O(n/\eps)$ in the journal version~\cite{GGL+00}. Fischer et al.~\cite{FLN+02} established the first lower bounds shortly after, showing that there exists a constant distance parameter $\eps_0 > 0$ such that $\Omega(\log n)$ queries are necessary for any \emph{non-adaptive} tester (one whose queries do not depend on the oracle's responses to prior queries).  This directly implies an $\Omega(\log \log n)$ lower bound for adaptive testers, since any $q$-query adaptive tester can be simulated by a non-adaptive one that simply carries out all $2^q$ possible executions.
\green{(Via a simple argument \cite{FLN+02} 
also gave an $\Omega(n^{1/2})$ lower bound for non-adaptive \emph{one-sided} testers, which must 
output {\sf Yes} with probability 1 if $f$ is monotone.  Throughout this work we consider only 
general two-sided testers, for which lower bounds are more difficult to prove.)}

In spite of considerable work on this problem and its variants \cite{GGL+98, DGL+99, GGL+00, FLN+02, AC06, HK08, BCGM12}, these were the best known results for the basic problem for more than a decade, until Chakrabarty and Seshadhri~\cite{CS13a} improved on the linear upper bound of Goldreich et al.~with an $\tilde{O}(n^{7/8}\eps^{-3/2})$-query tester. More recently, Chen et al.~\cite{CST14} closed the gap between upper and lower bounds on the query complexity of non-adaptive testers to within a polynomial factor by giving a lower bound of $\tilde{\Omega}(n^{1/5})$ (an exponential improvement of the~\cite{FLN+02} lower bound).  \cite{CST14} also gave an upper bound of $\tilde{O}(n^{5/6}\eps^{-4})$ queries (a polynomial improvement of the~\cite{CS13a} upper bound in terms of the dependence on $n$).

In this paper we make further progress towards a complete resolution of the problem with a lower bound of (almost) $\Omega(n^{1/2})$ against non-adaptive testers, which we conjecture
is optimal.
In more detail, our main result is the following:

\begin{theorem} \label{thm:main}
For all $c > 0$ there is a $\red{\kappa} = \red{\kappa}(c)>0$ such that
any non-adaptive algorithm for testing whether $f: \{-1,1\}^n \to \{-1,1\}$ is
monotone versus $\red{\kappa}$-far from monotone must use $\Omega(n^{1/2 - c})$ queries.
\end{theorem}


The paper of Chen et al.~\cite{CST14} also considered the problem of testing monotonicity of Booelan-valued functions over general hypergrid domains $\{1,\ldots,m\}^n$ for $m\ge 2$, and showed that it reduces to that of testing monotonicity Boolean functions as defined above (i.e.~the case when $m=2$) with essentially no loss in parameters. More precisely, they proved that any lower bound for \red{$\kappa$}-testing monotonicity of $f\isafunc$ translates into a lower bound for $\red{\Omega(\kappa)}$-testing monotonicity of $F:\{1,\ldots,m\}^n\to\bits$ with only a logarithmic loss in terms of $n$ in the query lower bound. Therefore Theorem~\ref{thm:main} along with this reduction yields our most general result:

\begin{theorem}
For all $c > 0$ there is a $\red{\kappa} = \red{\kappa}(c) > 0$ such that for all $m\ge 2$, any non-adaptive algorithm for testing whether $F:\{1,\ldots,m\}^n \to\bits$ is monotone versus $\red{\kappa}$-far from monotone must use $\Omega(n^{1/2-c})$ queries.
\end{theorem}

\ignore{
%
}

\subsection{Previous work:  the \cite{CST14} lower bound} \label{sec:prevwork}

In order to explain our approach in the current paper we first briefly recall the key elements of the~\cite{CST14} lower bound.
That paper uses Yao's method, i.e.~it exhibits two
distributions $\calD_{yes},\calD_{no}$ over Boolean functions, where each $\boldf \sim \calD_{yes}$ is monotone
and almost every $\boldf \sim \calD_{no}$ is constant-far from monotone.  The main conceptual novelty of the \cite{CST14} lower bound was
to use linear threshold functions (LTFs) as both the yes- and no- functions, thereby enabling the application of sophisticated
multidimensional central limit theorems to establish the closeness in distribution that is required by Yao's method.
In more detail, a function drawn from the  ``yes-distribution'' $\calD_{yes}$ of \cite{CST14} is
\begin{equation} \label{eq:yes}
\boldf(X)=\sign (\bu_1 X_1 + \cdots + \bu_n X_n)
\end{equation}
where each $\bu_i$ is independently uniform over $\{1,3\}$, 
and a function drawn from the ``no-distribution'' $\calD_{no}$ is
\begin{equation} \label{eq:no}
\boldf(X) = \sign (\bv_1 X_1 + \cdots + \bv_n X_n)
\end{equation}
where each $\bv_i$ is independently $-1$ with probability $1/10$ and is $7/3$ with probability $9/10.$

Fix an arbitrary (adversarially chosen) $d \times n$ query matrix $\cal{X}$ whose elements all are  $ \pm 1/\sqrt{n}$, and let
$\calX^{(1)},\ldots,\calX^{(n)} \in\{\pm 1/\sqrt{n}\}^d$ be the columns of this matrix. The $d$ rows of this matrix correspond to an arbitrary $d$-element
set of $n$-bit query strings scaled by a factor of $1/\sqrt{n}$. (Note that scaling the input does not change the value of a zero-threshold linear threshold function such as (\ref{eq:yes}) or (\ref{eq:no}) above.)  Define the $\R^d$-valued random variables
\begin{equation}
\bS = \sum_{i=1}^n \bu_i \calX^{(i)} \quad\text{and}\quad \bT = \sum_{i=1}^n \bv_i \calX^{(i)}.\label{eq:SandT}
\end{equation}
Recalling (\ref{eq:yes}) and (\ref{eq:no}) and Yao's minimax lemma, to prove a $d$-query monotonicity testing lower bound for non-adaptive algorithms,
it suffices to upper bound
\begin{equation} \label{eq:duobound}
\duo(\bS,\bT) \leq 0.1
\end{equation}
(here the ``0.1'' constant is arbitrary, any constant in $(0,1)$ would do)
for all possible choices of $\calX$,
where $\duo$ is the ``union-of-orthants'' distance:
\[ \duo(\bS,\bT)  := \max\Big\{ |\Pr[\bS \in \calO]-\Pr[\bT\in\calO]| \colon \text{$\calO$ is a union of orthants in $\R^d$\Big\}}. \]
Thus, in this approach, the goal is to make $d$ be as large as possible (as a function of $n$) while keeping $\duo(\bS,\bT)$ at most 0.1.

To obtain their main $\tilde{\Omega}(n^{1/5})$ lower bound, \cite{CST14}  use a multidimensional central limit theorem (CLT) of
Valiant and Valiant \cite{VV11}, which is proved using Stein's method and which bounds
the earthmover (Wasserstein) distance between sums of independent vector-valued random variables.  \cite{CST14} adapts this earthmover
CLT to obtain a CLT for  the ``union-of-orthants'' distance $\duo$, and shows that using this CLT the value of $d$ can be taken as large as
$\tilde{\Omega}(n^{1/5})$.

The key properties of the random variables $\bu_i$ and $\bv_i$ used in \cite{CST14} are that

\begin{enumerate}

\item Their first and second moments match, i.e. $\E[\bu_i]=\E[\bv_i]$ and
$\E[\bu_i^2]=\E[\bv_i^2]$.  (This ensures that $\bS$ and $\bT$
have matching means and covariance matrices, which makes it possible
to apply the  \cite{VV11} CLT.)

\item The random variable $\bu_i$ is supported entirely on non-negative values, while $\bv_i$ has nonzero weight
on negative values.  (The first condition ensures that $\boldf \sim \calD_{yes}$ will be monotone, and the second ensures
that a random $\boldf \sim \calD_{no}$ will w.h.p. be constant-far from monotone.)

\end{enumerate}

\subsection{Our approach and techniques} \label{sec:techniques}

In light of the above, it is natural to ask whether imposing stronger requirements on the $\bu_i,\bv_i$ random variables
can lead to stronger results:  in particular, can matching higher moments than just the first two lead to an
improved lower bound?
Pursuing such an approach, one quickly discovers that
extending the \cite{VV11} CLT for earthmover distance (which, as mentioned above,
is proved using Stein's method) to exploit matching higher moments is a nontrivial technical challenge.
Instead, in this work we return to a much older proof method for CLTs, namely Lindeberg's ``replacement method''
(discussed in detail in Section~\ref{sec:lindeberg}), which is well suited for higher moments.
Our arguments show that
by combining a careful construction of the random variables (the coefficients of the LTFs) with a careful analysis of
all possible query matrices, the Lindeberg method can be used to obtain an $\Omega(n^{1/2 - c})$ lower bound for monotonicity
testing.

\red{We observe that a high-level difference between our paper and that of \cite{VV11} is that \cite{VV11} proves that a sum of independent $d$-dimensional random variables converges to a multi-dimensional Gaussian with matching first two moments (mean and covariance).
In contrast, we work with two different but carefully constructed sums of independent $d$-dimensional random variables
which have many matching moments, namely the $\bS$ and $\bT$ random variables defined in (\ref{eq:SandT}).  Our goal is \emph{not}
to establish smaller distance to a multi-dimensional Gaussian (indeed our arguments do not establish this); rather, as described above,
having $\duo(\bS,\bT) \leq 0.1$ is sufficient for our
purposes, and our goal is to achieve such ``rough'' closeness for $d$-dimensional random variables where $d$ is as large as possible
(i.e. as close as possible to $n^{1/2}$).}

As a warmup, in Section \ref{sec:lindeberg} we first prove an $\Omega(n^{1/4 - c})$ lower bound via a fairly
straightforward application of the Lindeberg method.  This argument essentially requires only matching moments of order
$1,2,\dots,\red{1/c}$ for the $\bu_i,\bv_i$ random variables without other special properties \red{--- in particular, it does not matter
just what those moments are as long as they match each other ---}
and the analysis proceeds in the usual way for the Lindeberg method.  However,
improving this lower bound to $\Omega(n^{1/2 - c})$ requires \red{many} new ideas and significantly more care in the construction and
analysis.  We discuss several of the necessary ingredients, and in so doing
give an overview of our proof approach, below.

\medskip

\noindent {\bf (1):  Suitable choice of distributions.}
We show that given any positive integer $\red{\ell}$, there is a non-negative value $\mu=\mu(\red{\ell})$ and a non-negative random variable $\bu$ such that the first $\red{\ell}$ moments of $\bu$ match those of
the mean-$\mu$, variance-1 Gaussian $\calN(\mu,1).$  (This non-negative support of $\bu$
ensures that the $\calD_{yes}$ functions defined by (\ref{eq:yes}) are monotone
as required.)  For the $\calD_{no}$ functions, we show that there is a
random variable $\bv$ (see (\ref{eq:no})) that has first $\red{\ell}$ moments
matching those of $\calN(\mu,1)$, has finite support, and takes negative values with nonzero probability.  The finite support and negativity conditions enable us to argue that
almost all functions drawn from $\calD_{no}$ are
indeed constant-far from monotone, and the fact that $\bv$'s moments match those of
a Gaussian plays a crucial role in enabling step (4) to go through, as described below.
\ignore{\rnote{Commented out a sentence ``The fact that $\bu$ and $\bv$ have matching first $\red{\ell}$ moments is crucial for the
Lindeberg-based approach to go through.'' since it felt redundant at this point in the exposition.}}

\medskip

\noindent {\bf (2):  Careful choice and analysis of mollifier.}  The Lindeberg method uses smooth ``mollifiers'' with useful analytic properties (bounded
derivatives and the like) to approximate discontinuous indicator functions.  We give a careful construction of a particular mollifier which
exploits some of the ``nice structure'' of the sets (unions of orthants) that we must deal with, and show how this mollifier's special
properties can be used to obtain a significant savings in bounding the error terms that arise in Lindeberg's method.  Our analysis based on
this particular mollifier shows that to bound the error terms in Lindeberg's method, it is enough to give an \emph{anticoncentration}
bound. \red{In more detail, we identify a family of (roughly) $d^{\mp+1}$ random variables $\bR_{-i}|_{J}$ (corresponding to the different
possible outcomes of the multi-index $J$ in (\ref{eq:useme}); see Section \ref{sec:goingbeyond}), and show that it is enough to establish that
for almost all of these random variables (outcomes of $J$), there is a strong upper bound on the probability that $\bR_{-i}|_{J}$ (which is a sum
of $\green{n-1}$ independent $(\mp+1)$-dimensional vector-valued random variables) lands in a small origin-centered rectangular box, which we
denote $\red{\calB}_J$, in $\R^{\mp+1}.$}
\blue{Here $\mp$ is a value which is chosen to be significantly
less than $\ell$, but still ``large enough'' that it suffices for
steps (2) and (3) being described here; we will use the
remaining $\ell - \mp$ matching moments later in the argument, in step (4).}
\ignore{
}

\medskip

\noindent {\bf (3):  Pruning arbitrary query sets.}  We may associate
each multi-index $J$ that has $|J|=\mp+1$ with a multiset $\orange{\calA}$ of size
$\red{\mp+1}$ drawn from the $d$-element query set.
For simplicity, in the following informal discussion let us assume that
every element in $\orange{\calA}$ occurs with multiplicity exactly $1$ (this is indeed
the case for most multisets of $[d]$ of size $\mp+1$; recall that $\mp$
is a fixed integer whereas $d$ should be thought of as $n^{\Theta(1)}$).
\ignore{\red{An outcome of the multiset $A$ corresponds to an outcome of the multi-index $J$ mentioned above.}
}
\ignore{Let $\bS_A$ be the $(\mp+1)$-dimensional random variable obtained by restricting
$\bS$ to the indices in $A$, and likewise define $\bT_A$.
In order to put strong upper bounds on the probability of $\bS_A$ and $\bT_A$ landing in a small rectangular box in $\R^C$,}

\red{A major difficulty is that for some query sets, it may be the case that for many outcomes of $\orange{\calA}$ (equivalently, $J$) it is simply impossible to give a strong upper bound on the probability that $\bR_{-i}|_{J}$  lands in the small rectangular box $\red{\calB_J}$. For example, this can be the case if many query strings lie very close to each other (see the discussion in the last two
paragraphs of Section \ref{sec:warmup} for an extreme instance of this phenomenon).  However, if there are two query strings which are very close to each other, then with very high probability over the outcomes of  $\bu_1,\dots,\bu_n$ the responses to the two queries
for $\boldf \sim \calD_{yes}$ will be the same, and likewise for $\boldf \sim \calD_{no}.$ This should effectively allow us to ``prune'' the query set and reduce its size by $1$. On the other hand, there is a non-zero probability that two close but distinct query strings have different answers, and it is intuitively clear that this probability increases with the distance between the query strings; thus any such pruning must be done with care.

There is indeed a delicate balance between these two competing demands (pruning queries to eliminate cases where the desired
anti-concentration probability cannot be effectively bounded, and introducing errors by pruning queries).  In Section~\ref{sec:pruning}, we perform a careful tradeoff between these demands, and show that \emph{any} query set can be pruned (at the cost of a small acceptable increase
in error) in a useful way.  The exact condition we require of our pruned query sets is rather involved so we defer a precise statement
of it to Section \ref{sec:pruning}, but roughly speaking, it involves having only a small fraction of all queries lie too close to the linear
span of any small set of query strings
(see Definition \ref{def:scattered} for a precise definition).
We show in later sections that this condition,
which we refer to as a query set being ``scattered,''
lets us establish the desired
anti-concentration mentioned above.}  
\ignore{
show how one can effectively prune the query set so that for any query, the fraction of queries at $\ell_2$ distance $r$ is $r \cdot \poly \log n$. \lnote{I agree that this is the main intuition, though our notion of scatteredness is significantly more delicate than that and I wonder if we should try to get that across.}  Note that after our normalization by $\sqrt{n}$, $r$ is a number between $0$ and $1$. This is not an accurate statement but for the sake of simplicity of discussion, we allow ourselves some level of imprecision.}We note that  our pruning procedure  heavily uses the fact that query strings are elements of the (scaled) Boolean hypercube; this enables us to establish and employ some useful facts which,
roughly speaking, exploit some geometrical incompatibility between linear subspaces of $\R^n$ and the Boolean hypercube.

\medskip

\noindent {\bf (4):  Handling \green{scattered} query sets.}
A careful analysis of \green{scattered}
query sets lets us show that if $\bG$ is a $\red{(\mp+1)}$-dimensional \emph{Gaussian} whose mean and
covariance matrix match those of $\bR_{-i}|_{J}$, then $\bG$ satisfies the desired anti-concentration bound.  To show that the above-mentioned
random variable $\bR_{-i}|_{J}$ --- which is \emph{not} a Gaussian --- also satisfies this anti-concentration bound, we exploit the fact that $\bu$'s
first $\green{\ell}$ moments match those
of $\bv$,  \emph{which in turn match the first $\red{\ell}$ moments of a variance-1 Gaussian} (recall ingredient (1), ``Suitable choice of distributions,'' above).
\blue{(This is where we use the ``remaining'' $\ell-h$ matching moments
for $\bu$ and $\bv$ alluded to earlier.)}
This lets us adapt the simple argument that was employed for the
``warm-up'' result to establish that the two distributions $\bR_{-i}|_{J}$ and $\bG$ must both put almost the same
amount of weight as each other on the box $\red{\calB}_{\green{J}}$ mentioned above; since $\bG$ is anti-concentrated on this box, it follows that
$\bR_{-i}|_{J}$ must have similar anti-concentration.

The fact that $\bv$ matches the first $\red{\ell}$ moments of a Gaussian is crucial here, since otherwise
the penalty incurred for the ``smoothing'' term in Lindeberg's method (the final term on the RHS of the inequality of Proposition \ref{bentkus-2.1})
would be prohibitively large.  By having $\bv$ match the moments of a Gaussian, though, we can use the aforementioned analysis (showing that the Gaussian
$\bG$ satisfies the desired anti-concentration bound) in order to give a strong upper bound on this
smoothing penalty, and thereby obtain our overall desired result.

\subsection{Organization.}

In Section \ref{sec:distributions} we establish the existence of real random variables $\bu,\bv$ with the ``matching moments'' property
that we require.  Section \ref{sec:warmup} proves an $\Omega(n^{1/4 - c})$ lower bound for monotonicity testing via a ``vanilla''
application of Lindeberg's method using higher-order matching moments, and outlines our approach
for going beyond $n^{1/4-c}.$  Section \ref{sec:pruning} describes our pruning procedure that transforms an arbitrary
query set into a ``scattered'' query set.  In Sections \ref{sec:lbscattered} we give our lower bound for
scattered query sets, and finally in Section \ref{sec:puttogether} we put together the pieces and complete the proof of Theorem~\ref{thm:main}.

\section{Preliminaries}

Given $n\in \N$, we let $[n]$ denote $\{1,\ldots,n\}$, and given $a \leq b \in \N$ we let
  $[a:b]$ denote $\{a,\dots,b\}.$~We use lowercase letters to denote real numbers, uppercase letters to denote vectors of real numbers, and boldface (e.g. $\bx$ and $\bX$) to denote random variables. We will also use calligraphic letters like $\cal X$ to
  denote sets or multisets of vectors.
  \ignore{
  }

For $X \in \R^n$ we use
$B_{\ell_2}(X,r)$ to denote $\{\hspace{0.015cm}Y \in \R^n: \|X-Y\|_2 \leq r
  \hspace{0.01cm}\}$,
the Euclidean ball of~radius $r$ centered at $X$.
For $Y,Z \in \{\pm 1/\sqrt{n}\}^n$, the Hamming distance $\dham(Y,Z)$
  is defined as the number of coordinates where $Y$ and $Z$ differ.


Recall that a $k$-variable Boolean function $f$ is a \emph{linear threshold function} (LTF) if
there exist real values $w_1,\dots,w_k,\theta$ such that $f(x) = \sign(\sum_{i=1}^k w_i x_i - \theta).$

We will require the following useful fact on the number of distinct LTFs over the
  $k$-dimensional Boolean hypercube (where
we view two LTFs as distinct if they differ as Boolean functions):
\begin{fact}\emph{\label{fact:number-halfspaces}\cite{Sch:50}}
 The total number of distinct LTFs over $\{-1,1\}^k$ is upper bounded by $2^{k^2}$.
\end{fact}

Given a $d$-dimensional multi-index $J = (J_1,\ldots,J_d) \in \N^d$, we write $|J|$ to denote $J_1 + \cdots + J_d$ and $J!$ to denote $J_1!J_2! \cdots J_d!$. We write $\supp(J)$ to denote the set $\{ i \in [d]\colon J_i \ne 0\}$, and $\#J$ to denote $|\supp(J)|$. (Note that $\#J \le |J|$.)
Given $X \in \R^d$ we write $X^J$ to denote $\prod_{i=1}^d (X_i)^{J_i}$, and $X|_J \in \R^{\#J}$ to denote the projection of $X$ onto the coordinates in $\supp(J)$. {For $f: \R^d \to \R$, we write $f^{(J)}$ to denote the $J$-th derivative,
i.e.
\[
f^{(J)} =
{\frac {\partial^{J_1 + \cdots + J_d} f}
{\partial x_1^{J_1} \cdots \partial x_d^{J_d}}}.
\]
}

We will use the standard multivariate Taylor expansion:

\begin{fact}[Multivariate Taylor expansion]\label{fact:taylor}
Given a smooth function $f:\R^d\to \R$ and $k\in \mathbb{N}$,
$$
f(X+\Delta)=\sum_{|J|\le k}\frac{f^{(J)}(X)}{J!}\cdot \Delta^J
+(k+1)\sum_{|J|=k+1}\left(\frac{\Delta^J}{J!}
\E\big[(1-\btau)^k f^{(J)}(X+\btau \Delta)\big]\right),
$$
for $X,\Delta\in \R^d$, where $\btau$ is a random variable
  uniformly distributed on the interval $[0, 1]$.
\end{fact}

We recall the standard Berry--Ess\'een theorem (see for example, \cite{Fel68})
for sums of independent real random variables:
\begin{theorem}[Berry--Ess\'een]
\label{thm:be}
Let $\bs = \bx_1 + \cdots + \bx_n$, where $\bx_1,\ldots,\bx_n$ are independent real-valued random variables with $\E[\bx_j] = \mu_j$ and $\Var[\bx_j] = \sigma_j^2$, and suppose that $| \bx_j - \E[\bx_j]|\le \tau$ with~proba\-bi\-lity $1$ for all $j\in [n]$. Let 
$\bg$ denote a Gaussian random variable with mean $\sum_{j=1}^n \mu_j$ and variance $\sum_{j=1}^n\sigma_j^2$, matching those of $\bs$. Then for all $\theta\in \R$, we have
\[ \big| \Pr[\bs \le \theta] - \Pr[
\bg \le \theta] \big| \le \frac{O(\tau)}{\sqrt{\sum_{j=1}^n \sigma_j^2}}.\]
\end{theorem}

\section{The $\calD_{yes}$ and $\calD_{no}$ distributions} \label{sec:distributions}

The main results of this section are the following:
\ignore{
}

\begin{proposition}[The ``yes'' random variable]
\label{prop:matchmoments}
Given an odd $\ell \in \N$, there exists a value $\mu=\mu(\ell)$ $> 0$
  and a real random variable $\bu$ such that\vspace{-0.04cm}
\begin{enumerate}
\item $\bu$ is supported on at most $\ell$ nonnegative real values; and\vspace{-0.12cm}
\item  $\E[\bu^k] = \E[\calN({\mu},1)^k]$ for all $k \in [\ell]$.\vspace{0.06cm}
\end{enumerate}
\end{proposition}

\begin{proposition}[The ``no'' random variable]
\label{prop:negsupport}
Given $\mu > 0$ and $\ell\in \N$, 
there exists a real random variable $\bv$ such that\vspace{-0.04cm}
\begin{enumerate}
\item $\bv$ is supported on at most $\ell+1$ real values, with $\Pr[\bv<0]>0$; and\vspace{-0.12cm}
\item $\E[\bv^k] = \E[\calN({\mu},1)^k]$ for all $k \in [\ell]$.\vspace{0.06cm}
\end{enumerate}
\end{proposition}

\noindent Note the difference between these two propositions:  the first requires $\bu$ to be supported entirely~on nonnegative values, while the second requires
  $\bv$ to put nonzero weight on some negative value.

Let $c > 0$ (this should be viewed as the ``$c$'' of Theorem~\ref{thm:main}), and 
let $\mp=\mp(c)\in \N$ denote an odd constant that depends on $c$ only.
Let $\bu$ and $\bv$ denote random variables
  given in Proposition \ref{prop:matchmoments} and \ref{prop:negsupport}, respectively,
  with $\ell=\mp^3$ and $\mu=\mu(\ell)$.
As discussed~in Section \ref{sec:prevwork} the ``yes'' distribution $\calD_{yes}$
  of Boolean functions is given by
  (\ref{eq:yes}) and the ``no'' distribution by (\ref{eq:no}),
   where each $\bu_i$ is i.i.d. distributed according to $\bu$
  and likewise for the $\bv_i$'s and $\bv$.
It is clear that $\bu$ and $\bv$ have matching first $\ell$-th moments, and Pro\-position \ref{prop:matchmoments} ensures that every function~in the support of
$\calD_{yes}$ is monotone.  In Appendix \ref{sec:nofar} we show that with probability $1-o_n(1)$, a random LTF drawn from $\calD_{no}$ is
$\red{\kappa}$-far from all monotone Boolean functions, where $\red{\kappa}>0$ depends
  on the values of $\mu$ and $\ell$ and hence on $c$ only.

Thus the above two Propositions \ref{prop:matchmoments} and \ref{prop:negsupport}
are enough for the basic framework
of Yao's method to go through and establish our lower bound, once we show
that equation (\ref{eq:duobound}) holds.  We do this in the rest of the paper, but first
in the remainder of this section we prove Propositions \ref{prop:matchmoments} and \ref{prop:negsupport}.
We start with the easier Proposition \ref{prop:negsupport}.

\subsection{Proof of Proposition \ref{prop:negsupport}}

For each $x \in \R$, let $\blue{A}{(x)}$ denote the $(\ell+1)$-dimensional
real vector defined by
\[
A{(x)}_k =
\begin{cases}
x^k 		& \text{~for~}k\in [\ell],\\
\ind[x<0] 	& \text{~for~}k=\ell+1.
\end{cases}
\]
Consider the vector $\blue{P} \in \R^{\ell+1}$ defined by
\[
P_k =
\begin{cases}
\E[{\cal N}(\mu,1)^k] 		& \text{~for~}k\in [\ell],\\[0.2ex]
\Pr[{\cal N}(\mu,1)<0] 		& \text{~for~}k=\ell+1.
\end{cases}
\]
Since $P = \E_{\bx \sim {\cal N}(\mu,1)}[A{(\bx)}]$ the point
$P$ is in the convex hull of the point set
$V := \{\blue{A(x)}: x \in \R\} \subset \R^{\ell+1}$.
Hence Carath\'eodory's theorem implies that $P$ lies in the convex hull
of some $(\ell+1)$-point subset of $V$, i.e. there exist $x_1,\dots,x_{\ell+1}
\in \blue{\R}$ and $0 \leq \mu_1,\dots,\mu_{\ell+1}$ with $\sum_j \mu_j=1$
such that
$$P = \sum_{j=1}^{\ell+1} \mu_j A{(x_j)}.$$

The desired random variable $\bv$ is defined by $\Pr[\bv = x_j]=\mu_j.$
It is clear that $\bv$ is supported on at most $\ell+1$ real values, and
$\bv$ satisfies the desired moment condition since
\[
\E[{\cal N}(\mu,1)^k] = P_k = \sum_{j=1}^{\ell+1} \mu_j x_j^k
= \E[\bv^k],\ \ \ \text{for all~}k\in [\ell].
\]
Finally, since
\[
0 < \Pr[{\cal N}(\mu,1)<0] = P_{\ell+1} = \sum_{j=1}^{\ell+1} \mu_j \ind[x_j<0],
\]
it must be the case that $\Pr[\bv<0]>0$, and the proposition is proved.
\qed

\subsection{Proof of Proposition \ref{prop:matchmoments}}

Given a sequence of real numbers $(m_1, \ldots, m_\ell)$, consider
the problem of deciding whether there~exists
a real random variable $\bx$ such that $\mathbf{E}[\bx^i]=m_i$ for $i=1,\dots,\ell.$
This is a form of the well-studied classical moment problem, and a complete solution
has been given in terms of the moment vector lying in a particular
well-specified cone. More precisely, the
following can be found in~\cite{Akh65}.

\begin{theorem}\label{thm:moment-feasibility-A}
 Let $\overline{m}=(m_1, m_2, \ldots, m_{2n})$. There is a random variable $\bx$ supported on $\R$ such that $\mathbf{E}[\bx^i]=m_i$ for $i=1,\ldots, 2n$ if and only
 if
 \[ A_{\R}(\overline{m}) =
\begin{pmatrix}
1 & m_1 & \ldots & m_n \\
m_1 & m_2 & \ldots & m_{n+1} \\
\vdots & \vdots & \ddots & \vdots \\
m_n & m_{n+1} & \ldots & m_{2n} \end{pmatrix}  \succeq 0. \]
\end{theorem}
The corresponding problem when the support of the desired random variable $\bx$ is restricted to non-negative
  reals is also completely solved, by the following result:
\begin{theorem}\label{thm:moment-feasibility-B}
 Let $\overline{m}=(m_1, m_2, \ldots, m_{2n+1})$. \hspace{-0.04cm}There is a random variable $\bx$ supported on $[0,\infty)$~such that $\mathbf{E}[\bx^i]=m_i$ for $i=1,\ldots, 2n+1$ if and only if
$$
A_{\R}(\overline{m})=
\begin{pmatrix}
1 & m_1 & \ldots & m_n \\
m_1 & m_2 & \ldots & m_{n+1} \\
\vdots & \vdots & \ddots & \vdots \\
m_n & m_{n+1} & \ldots & m_{2n} \end{pmatrix}  \succeq 0\ \ \
\textrm{and}\ \ \ A_{\R}^+(\overline{m})=\begin{pmatrix}
m_1 & m_2 & \ldots & m_{n+1} \\
m_2 & m_3 & \ldots & m_{n+2} \\
\vdots & \vdots & \ddots & \vdots \\
m_{n+1} & m_{n+2} & \ldots & m_{2n+1} \end{pmatrix} \succeq 0.
$$
\end{theorem}

We use the above results by taking each $m_\ell$ to equal  $\E[\bz_\mu^{\ell}]$,
where $\red{\bz_\mu}$ is distributed according to $\calN(\mu,1)$ (so $m_\ell=m_\ell(\mu)$ is
a function of $\mu$).
Our aim is to show that
$\mu=\mu(\ell)$ can be taken to be a sufficiently large integer (in terms of $\ell$) such that
\begin{eqnarray} \label{eq:A}
&&A_{\R}(\overline{m}(\mu))=\begin{pmatrix}
1 & m_1 & \ldots & m_\ell \\
m_1 & m_2 & \ldots & m_{\ell+1} \\
\vdots & \vdots & \ddots & \vdots \\
m_\ell & m_{\ell+1} & \ldots & m_{2\ell} \end{pmatrix} \succeq 0\\[0.6ex]
 \text{and}
&&A_{\R}^+(\overline{m}(\mu)) = \begin{pmatrix}
m_1 & m_2 & \ldots & m_{\ell+1} \\
m_2 & m_3 & \ldots & m_{\ell+2} \\
\vdots & \vdots & \ddots & \vdots \\
m_{\ell+1} & m_{\ell+2} & \ldots & m_{2\ell+1} \end{pmatrix} \succeq 0. \label{eq:B}
\end{eqnarray}
If these two conditions hold, then we may take $\bu'$ to be the nonnegative random
variable $\bx$ whose existence~is asserted by Theorem~\ref{thm:moment-feasibility-B}.
Applying Carath\'eodory's theorem, an argument similar to
  the proof of Proposition \ref{prop:negsupport}
  allows us to obtain from $\bu'$ a nonnegative random variable $\bu$ with
  support size at most $\ell$ and the same moments.
This will finish the proof of Proposition \ref{prop:matchmoments}.

As the Gaussian $\red{\bz_\mu}=\calN(\mu,1)$ is itself a random variable such that $\E[\bz_\mu^i]=m_i$,
for $i=1,\ldots,2\ell$, Theorem~\ref{thm:moment-feasibility-A} implies that (\ref{eq:A}) holds; thus,
it remains to prove (\ref{eq:B}).

Observe that $m_{k}=m_k(\mu)$ is a degree-$k$ polynomial in $\mu$.
We define
\[
P_{\det}(\mu)=\det(A_{\R}^+(\overline{m})).
\]
Our argument requires the following four technical claims.

\begin{claim} \label{claim:nonsingular}
There exists $\mu_0$ such that  $A_{\R}^+(\overline{m}(\mu))$ is non-singular for all $\mu \in \R \setminus [-\mu_0, \mu_0]$.
\end{claim}
\begin{proof}
Observe that $P_{\det}(\mu)$
  is a degree-$T$ polynomial in $\mu$ with $T<(\ell+1)(2\ell+1)$.  
Thus, given that $P_{\det}(\mu)$ is not the identically-0 polynomial, if $\mu_0$ is set to be the
  largest magnitude of the zero of this polynomial, we get the claim.

To see that $P_{\det}(\mu)$ is not identically zero,
we consider the matrix $A_{\R}^+(\overline{m}(0))$
obtained by taking $\mu=0$.  This matrix has $(i,j)$th entry $\E[{\cal N}(0,1)^{i+j-1}]$, which is 0 if $i+j$ is even. 
By inspection of this matrix we see that for odd $\ell>1$,
  we have\vspace{-0.06cm}
$$
A_{\R}^+(\overline{m}(0))=\begin{pmatrix} 0&1\\1&0\end{pmatrix}\otimes
B^{(\ell)},\vspace{-0.06cm}
$$
where $B^{(\ell)}$ is the square matrix
of dimension $(\ell+1)/2$ that has $(i,j)$th entry $\E[{\cal N}(0,1)^{2(i+j-1)}].$
It follows that $P_{\det}(0) =
(-1)^{(\ell+1)/2}\det(B^{(\ell)})^2$.
Recalling that for $k$ even we have
$$\E[{\cal N}(0,1)^{k}] = (k-1)!! = 1 \cdot 3 \cdot 5
\cdots (k-1),$$ the product of the odd numbers from $1$ to $k-1$, it can be shown
that
\[
\det(B^{(\ell)}) = \prod_{j \text{~odd}, 1 \leq j \leq \ell} j!.
\]
We include a proof of this fact in Appendix \ref{sec:det}. The lemma then follows.
\end{proof}
\begin{claim}\label{clm:detbound}
 If $\mu > \mu_0$ is an integer, then $|\hspace{-0.02cm}\det(A_{\R}^+(\overline{m}(\mu)))| \ge 1$.
\end{claim}
\begin{proof}
The fact that each raw moment $\E[{\cal N}(0,1)^k]$ of the Gaussian is an integer easily implies that
$m_k(\mu)=\E[({\cal N}(0,1)+\mu)^k]$ is a polynomial with integer coefficients, and hence
$P_{\det}(\mu)$ has integer coefficients as well.  Together with Claim \ref{claim:nonsingular} this gives the claim.\end{proof}

\begin{claim}\label{clm:singbound}
For all integer $\green{\mu > 0}$, we have that the largest singular value of
$A_{\R}^+(\overline{m}(\mu))$, denoted $\sigma_{\max}(A_{\R}^+(\overline{m}\red{(\mu)})),$
satisfies
$\sigma_{\max}(A_{\R}^+(\overline{m}\red{(\mu)})) \le  
  (\ell+1)^2\cdot (2\ell+1)!\cdot \mu^{2\ell+1}.$
\end{claim}
\begin{proof}
We have $\sigma_{\max}(A_{\R}^+(\overline{m}\red{(\mu)}))
\le \Vert A_{\R}^+(\overline{m}\red{(\mu)}) \Vert_F$.
We use the following simple upper bound on the $k$th
  moment of the mean-$\mu$, variance-$1$ Gaussian $\calN(\mu,1)$:
$$
\E[\calN(\mu,1)^k] = \E[({\cal N}(0,1) + \mu)^k] 
=\sum_{j=0}^{\lfloor k/2\rfloor} {k \choose 2j}
  (2j-1)!!\cdot \mu^{k-2j}
<
(\lfloor k/2\rfloor +1)\cdot k!\cdot \mu^{k}.
$$
The claim follows by combining the two inequalities.
\end{proof}
\begin{claim}
\red{For all integer $\mu > \mu_0$, we have}
that the smallest singular value of
$A_{\R}^+(\overline{m}(\mu))$, denoted $\sigma_{\min}(A_{\R}^+(\overline{m}\red{(\mu)})),$
satisfies
$$ \sigma_{\min}(A_{\R}^+(\overline{m}\red{(\mu)})) \ge
\frac{1}{(2\ell+1)^{\green{\ell(2\ell+2)}} \cdot \mu^{\green{\ell(2\ell+1)}}}.$$
\end{claim}
\begin{proof}
 We just use the simple inequality that for any symmetric matrix $A \in \R^{\green{(\ell+1) 
   \times (\ell+1)}}$:\ignore{\footnote{\green{Xi: This was $\R^{\ell\times \ell}$ before but I think it should be 
   $\ell+1$ instead. This also led to several changes marked by green.}}}
 $$
 \sigma_{\min}(A) \ge \frac{|\det(A)|}{\sigma_{\max}(A)^{\green{\ell}}}\ge
\frac{1}{(2\ell+1)^{\green{\ell(2\ell+2)}} \cdot \mu^{\green{\ell(2\ell+1)}}}.
 $$
 The last inequality uses Claim~\ref{clm:detbound} and Claim~\ref{clm:singbound}.
\end{proof}

With these technical claims in hand, we proceed to establish (\ref{eq:B}).
In case $A_{\R}^+(\overline{m}\red{(\mu)}) \succeq 0$ \red{for some integer $\mu > \mu_0$,}
we are done. Otherwise, towards a contradiction, let us assume that  $A_{\R}^+(\overline{m}\red{(\mu)})$ has a
negative eigenvalue \red{for every integer $\mu > \mu_0$}. This means that
\red{for every integer $\mu > \mu_0$, we have}
\begin{equation} \label{eq:lambdaminbound}
\lambda_{\min}(A_{\R}^+(\overline{m}\red{(\mu)})) \le
-\frac{1}{(2\ell+1)^{\green{\ell(2\ell+2)}} \cdot \mu^{\green{\ell(2\ell+1)}}}.
\end{equation}

Let us define the random variable $\bz_{\red{\mu}}'$ to be distributed as $\bz_{\red{\mu}}'=\max\{\bz_{\red{\mu}},0\}.$
It is straightforward to upper bound the difference in moments between the
random variables $\bz_{\red{\mu}}$ and $\bz'_{\red{\mu}}$:


\begin{claim}\label{clm:diffinmoments}
 For $k \in \mathbb{N}$, we have 
$\big|\mathbf{E}[\bz_{\red{\mu}}'^{k}] -\mathbf{E}[\bz_{\red{\mu}}^{k}] \big|
 \le e^{ -{\frac {\mu^2} {2}} } \cdot 
(k-1)!!.
 $
\end{claim}
\begin{proof} 
We have\vspace{-0.08cm}
$$\begin{aligned}
\big|\mathbf{E}[\bz_{\red{\mu}}'^{k}] -\mathbf{E}[\bz_{\red{\mu}}^{k}]\big|&=
\int_{y=-\infty}^{0}  \frac{1}{\sqrt{2\pi}} \cdot |y|^{k} \cdot e^{-\frac{(y-\mu)^2}{2}} dy \\[0.3ex]
&\le e^{-\frac{\mu^2}{2}}\int_{y=-\infty}^{0}  \frac{1}{\sqrt{2\pi}} \cdot |y|^{k} \cdot e^{-\frac{y^2}{2}} dy \\[0.5ex]
&= e^{ -{\frac{\mu^2} {2}} } \cdot (k-1)!! 
\end{aligned}$$
where the last line used the fact that the $k$-th absolute moment of $\calN(0,1)$ is
  at most $(k-1)!!$.
\end{proof}


To conclude the proof, let $m'_i\red{(\mu)}=\mathbf{E}[\bz_{\red{\mu}}'^i]$.
Then by Claim \ref{clm:diffinmoments}, we have \green{(where $\|M\|_\infty$
denotes the entrywise maximum absolute value of any element of the matrix $M$)} 
\begin{equation}
\label{eq:inf-bound}
\Vert A_{\R}^+(\overline{m}\red{(\mu)}) -A_{\R}^+(\overline{m}'\red{(\mu)}) \Vert_{\infty} \le
\ignore{(\ell+1)\cdot}
e^{-\frac{\mu^2}{2}}
\cdot
(2\ell)!!.
\end{equation}


\green{Let $u \in \R^{\ell+1}$ be the unit vector minimizing $u^{T}
A_{\R}^+(\overline{m}\red{(\mu)}) u,$ so $\lambda_{\min}(A_{\R}^+(\overline{m}\red{(\mu)})) = 
u^{T} A_{\R}^+(\overline{m}\red{(\mu)}) u.$  We have
\begin{eqnarray*}\begin{aligned}
\lambda_{\min}(A_{\R}^+(\overline{m}'\red{(\mu)})) - \lambda_{\min}(A_{\R}^+(\overline{m}\red{(\mu)}))
&\leq u^{T} (A_{\R}^+(\overline{m}'\red{(\mu)}) - A_{\R}^+(\overline{m}\red{(\mu)})) u \\
&\leq
(\ell+1) \|A_{\R}^+(\overline{m}'\red{(\mu)}) - A_{\R}^+(\overline{m}\red{(\mu)})\|_\infty,
\end{aligned}
\end{eqnarray*} 
}
so by (\ref{eq:lambdaminbound}) and (\ref{eq:inf-bound}) we get that
$$
\lambda_{\min}(A_{\R}^+(\overline{m}'\red{(\mu)})) \le
-\frac{1}{(2\ell+1)^{\green{\ell(2\ell+2)}} \cdot \mu^{\green{\ell(2\ell+1)}}}
+
(\ell+1)\cdot
e^{-\frac{\mu^2}{2}}
\cdot
(2\ell)!!.
$$
By choosing $\mu=\mu(\ell)$ to be a sufficiently large integer relative to $\ell$
  we can make
$\lambda_{\min}(A_{\R}^+(\overline{m'}\red{(\mu)})) <0$,
which is a contradiction with Theorem~\ref{thm:moment-feasibility-B}
  and the fact that $\bz'_{\red{\mu}}$ is supported on $[0,\infty)$.

\section{Warmup: an $\Omega(n^{1/4-c})$ lower bound via higher moments}
\label{sec:warmup}

In this section we give the basic Lindeberg argument using matching
higher moments.  This immediately improves the $\tilde{\Omega}(n^{1/5})$
lower bound in~\cite{CST14} to $\Omega(n^{1/4-c})$  for any
constant $c > 0$ (see the end of Section \ref{sec:lindeberg})
 and is the first step in our proof of the
$\Omega(n^{1/2 - c})$ lower bound.
The main technical ingredient is a higher-moments extension of
the~\cite{GOWZ10} multidimensional CLT, which we use in place of
the~\cite{VV11} multidimensional CLT used in~\cite{CST14}.



\subsection{A useful mollifier}

We begin with a couple of basic propositions:

\begin{proposition}\label{simplepro}
Let $\calA,\calA_{in}\sse \R^d$ where $\calA_{in} \sse \calA$. Let $\Psi_{in}: \R^d \to [0,1]$ be a function   satisfying $\Psi_{in}(X) = 1$ for all $X \in \calA_{in}$ and $\Psi_{in}(X) = 0$ for all $X \notin \calA$. Then for all random variables $\bS,\bT$:
\begin{equation*}
 \big|\hspace{-0.04cm}\Pr[\bS \in \calA]-\Pr[\bT\in \calA]
\hspace{0.01cm} \big| \le \big|\hspace{-0.04cm}\E[\Psi_{in}(\bS)]- \E[\Psi_{in}(\bT)]\hspace{0.01cm}\big|+
\max\hspace{-0.03cm}\big\{\hspace{-0.05cm}\Pr[\bS \in \calA\setminus \calA_{in}],\hspace{0.04cm}
\Pr[\bT\in \calA\setminus \calA_{in}]\big\}.\end{equation*}
\end{proposition}
\begin{proof}
Observe that $\Pr[\bS\in \calA] \ge \E[\Psi_{in}(\bS)]$  and $\Pr[\bS\in \calA] \le \E[\Psi_{in}(\bS)] + \Pr[\bS \in \calA\setminus \calA_{in}]$, and likewise for $\bT$. As a result, we have
$$\begin{aligned}
\Pr[\bS \in \calA]-\Pr[\bT\in \calA] &\le \E[\Psi_{in}(\bS)] + \Pr[\bS \in \calA\setminus \calA_{in}]-\E[\Psi_{in}(\bT)],\ \ \ \text{and}\\[0.3ex]
\Pr[\bS \in \calA]-\Pr[\bT\in \calA]  &\ge \E[\Psi_{in}(\bS)] -  \Pr[\bT \in \calA\setminus \calA_{in}]-\E[\Psi_{in}(\bT)].
\end{aligned}$$
Combining these, we have the proposition.
\end{proof}

We will use the following lemma of Bentkus~\cite{Ben03}.
For completeness we include its proof.

\begin{proposition}[Lemma 2.1 of \cite{Ben03}]
\label{bentkus-2.1}
Let $\calA,\calA_{in},\calA_{out} \sse \R^d$ with $\calA_{in} \sse \calA\sse \calA_{out}$. Let $\Psi_{in} : \R^d \to [0,1]$ be a function where $\Psi_{in}(X) = 1$ for all $X \in \calA_{in}$ and $\Psi_{in}(X) = 0$ for all $X \notin \calA$, and let $\Psi_{out} : \R^d \to [0,1]$ be a function where $\Psi_{out}(X) = 1$ for all $X \in \calA$ and $\Psi_{out}(X) = 0$ for all $X \notin \calA_{out}$.  Then for all random variables $\bS,\bT$ we have that
\begin{eqnarray*}\begin{aligned}
 \big|\hspace{-0.04cm}\Pr[\bS \in \calA]-\Pr[\bT\in \calA]\big| &\le \max\left\{\big|\hspace{-0.03cm}\E[\Psi_{in}(\bS)]- \E[\Psi_{in}(\bT)]\big|,\hspace{0.04cm} \big|\hspace{-0.02cm}\E[\Psi_{out}(\bS)]- \E[\Psi_{out}(\bT)]\big|\right\} \\[0.4ex]
 & \ \ \ \ + \max\big\{\hspace{-0.04cm}\Pr[\bT \in \calA_{out}\setminus \calA],\hspace{0.04cm} \Pr[\bT\in \calA\setminus \calA_{in}]\big\}.\end{aligned}\end{eqnarray*}
\end{proposition}
\begin{proof}
For the case when $\Pr[\bS \in \calA]\ge \Pr[\bT\in \calA]$, we have
$$
\Pr[\bS \in \calA]-\Pr[\bT\in \calA]\le \E[\Psi_{out}(\bS)]-\E[\Psi_{out}(\bT)]
  +\E[\Psi_{out}(\bT)]-\Pr[\bT\in \calA].
$$
The proposition follows from $\E[\Psi_{out}(\bT)]\le
  \Pr[\bT\in \calA_{out}]=\Pr[\bT\in \calA]+\Pr[\bT\in \calA_{out}\setminus \calA]$.

Now for the case when $\Pr[\bT\in \calA]>\Pr[\bS \in \calA]$, we have
$$
\Pr[\bT\in \calA]-\Pr[\bS \in \calA]\le
\Pr[\bT\in \calA]-\E[\Psi_{in}(\bT)]+\E[\Psi_{in}(\bT)]-\E[\Psi_{in}(\bS)].
$$
The proposition follows from $\Pr[\bT\in \calA]   \le \E[\Psi_{in}(\bT)]+\Pr[\bT\in \calA\setminus\calA_{in}]$.
\end{proof}

\def\ff{{\red{\alpha}}}

For the rest of this section, we need to define a sufficiently fast growing function
  $\ff: \mathbb{N} \rightarrow \R^+$:
$$
\ff (k) = 2e \cdot (64)^{k} \cdot k! \cdot \green{k^{2k+2}} .\ \ \ 
\ignore{\footnote{(Rocco:  The ``$k^{2k+2}$'' had been 
``$\max\{1,k^{2k+2}\}$.   I changed it in accordance with Xi's suggestion.  \green{Xi: It seems that 
  we don't need the $\max$ here? since $k\ge 1$. I think [KNW10] included it so that it gives
  a bound for the range of the mollifier itself but we stated that explicitly in all propositions.}}}
$$

As is standard in Lindeberg-type arguments,
our proof will employ a ``mollifier'', i.e. a smooth
function which approximates the indicator function of a set.  In this work
we require a specific mollifier whose properties are tailored to our sets of interest
(unions of orthants) and are given in
the following proposition.
\begin{proposition}[Product mollifier]
\label{product-mollifier}
Let $\calO$ be a union of orthants in $\R^d$. For all $\eps > 0$, there exists a smooth function $\Psi_\calO:\R^d\to [0,1]$ with the following properties:
\begin{flushleft}\begin{enumerate}
\item $\Psi_\calO(X) = 0$ for all $X \notin \calO$.\vspace{-0.06cm}
\item $\Psi_\calO(X) = 1$ for all $X \in \calO$ with $\min_{i}\{ |X_i|\} \ge \eps$.\vspace{-0.08cm}
\item For any multi-index $J \in \mathbb{N}^d$ such that $|J|=k$, $\Vert
\Psi^{(J)}_{{\cal O}} \Vert_{\infty} \le \ff(k) \cdot (1/\eps)^k$.\vspace{-0.12cm}
\item For any $J \in \mathbb{N}^d$,
$\Psi^{(J)}_{{\cal O}}(X) \not =0$ {only if} $X\in \calO$ and $|X_i| \le
\epsilon$ \vspace{0.02cm}for all $i$\vspace{-0.03cm} such that $J_i \ne 0$. Equivalently,
  $\Psi_{{\cal O}}^{(J)}(X) \ne 0$ {only if} $X\in \calO$ and $\| X|_J \|_\infty \le \eps$.
\end{enumerate}\end{flushleft}
\end{proposition}
We note that while properties (1)--(3) above are
entirely standard, we are not aware of previous work which uses property (4).  As we shall see
this property is particularly useful in our setting where the goal is to bound the union-of-orthants
distance $\duo$.
To prove Proposition~\ref{product-mollifier}, we first prove the following easier version of it.
\begin{proposition}\label{prop:prod-molli}
Let $\calO_1$ be an orthant in $\R^d$. For all $\eps > 0$,
  there exists a smooth function $\Psi_{\calO_1}:$
  $\R^d\to [0,1]$ with the following properties:
\begin{flushleft}\begin{enumerate}
\item $\Psi_{\calO_1}(X) = 0$ for all $X \notin \calO_1$.\vspace{-0.06cm}
\item $\Psi_{\calO_1}(X) = 1$ for all $X \in \calO_1$ with $\min_{i}\{ |X_i|\} \ge \eps$.\vspace{-0.08cm}
\item For any multi-index $J \in \mathbb{N}^d$ such that $|J|=k$, $\Vert
\Psi^{(J)}_{{\cal O}_1} \Vert_{\infty} \le \ff(k) \cdot (1/\eps)^k$.\vspace{-0.15cm}
\item For any $J \in \mathbb{N}^d$,
$\Psi^{(J)}_{{\cal O}_1}(X) \not =0$ {only if} $X\in \calO_1$ and $|X_i| \le
\epsilon$ for all $i$ such that \vspace{-0.03cm}$J_i \ne 0$. Equivalently,
  $\Psi_{{\cal O}_1}^{(J)}(X) \ne 0$ {only if} $X\in \calO_1$ and $\| X|_J \|_\infty \le \eps$
\end{enumerate}\end{flushleft}
\end{proposition}
We first see how Proposition~\ref{prop:prod-molli} can be used to prove Proposition~\ref{product-mollifier}.
\begin{proof}[Proof of Proposition~\ref{product-mollifier}]
Let ${\cal O} = \cup_{i \in [\blue{m}]} {\cal O}_i$, where the ${\cal O}_i$'s
  are (disjoint) orthants in $\R^d$. Let $\Psi_{{\cal O}_i}$ be the function
  obtained by applying Proposition~\ref{prop:prod-molli}
  to the orthant ${\cal O}_i$, and let $\Psi_{\calO} = \sum_{i \in [\blue{m}]} \Psi_{{\cal O}_i}$.
We claim that $\Psi_{\calO}$ satisfies the required conditions. Properties (1) and (2)
follow immediately from the corresponding properties of $\Psi_{{\cal O}_i}$.

For properties (3) and (4), observe that from Proposition~\ref{prop:prod-molli}, for each
$i \in [\blue{m}]$, $\Psi_{{\cal O}_i}^{(J)}(X)=0$ if $X \not \in {\cal O}_i$.
Also by the definition of $\Psi_{\calO}$ we have
$$
\Psi^{(J)}_{{\cal O}}(X) =\sum_{i \in [\blue{m}]}\Psi^{(J)}_{{\cal O}_i}(X).
$$
Since the ${\cal O}_i$'s are pairwise disjoint, we have that for any $X\in \R^d$,
  at most one of the summands~is non-zero. Thus, using property (3) from Proposition~\ref{prop:prod-molli}, we get
property (3) for  $\Psi_{\calO}$. Using the same reasoning and property (4) from Proposition~\ref{prop:prod-molli}, we get property (4) for $\Psi_{\calO}$.
\end{proof}

To prove Proposition~\ref{prop:prod-molli} we will need the following \emph{one-dimensional} version of $\Psi_{\calO_1}$. This is~the standard
mollifier construction in one-dimension. For completeness
we prove it in Appendix \ref{onedim}.
\begin{claim}\label{clm:smooth}
 For all $\epsilon>0$, there exists a smooth function $\Phi_{\epsilon} : \R \rightarrow [0,1]$ which satisfies:
 \begin{enumerate}
  \item If $x<0$, then $\Phi_\eps(x)=0$.\vspace{-0.1cm}
  \item If $x>\epsilon$, then $\Phi_{\epsilon}(x)=1$.\vspace{-0.1cm}
  \item $\Vert \Phi^{(k)}_{\epsilon} \Vert_\infty \le \ff(k) \cdot ({1}/{\epsilon})^k$.
 \end{enumerate}
\end{claim}

\begin{proof}[Proof of Proposition~\ref{prop:prod-molli}]
Without loss of generality we may assume our orthant $\mathcal{O}_1$ is $(\R^{+})^{d}$.
Let $$\Psi_{\calO_1}(X)=\prod_{i\in [d]} \red{\Phi}_{\epsilon}(X_i).$$
Then Properties (1) and (2) of Proposition~\ref{prop:prod-molli}
  follow directly from Properties (1) and (2) of Claim \ref{clm:smooth}.
By the product rule and the definition of $\Psi_{\calO_1}$, we have
  for any multi-index $J$:
\begin{equation}\label{eq:multi-derivative}
\Psi_{\calO_1}^{(J)}(X)=\prod_{i\in [d]} \Phi_{\epsilon}^{(J_i)}(X_i).
\end{equation}
Using property (3) of Claim~\ref{clm:smooth}, we get
$$
\Vert \Psi_{\calO_1}^{(J)} \Vert_{\infty} = \prod_{i\in [d]} \Vert \Phi_{\epsilon}^{(J_i)} \Vert_{\infty}
\le \prod_{i\in \supp(J)}\ff(J_i) \cdot (1/\epsilon)^{J_i} \le \ff(k) \cdot (1/\epsilon)^k.
$$
The last inequality uses that $\sum_{i=1}^d J_i=k$ and that $\log \ff( \cdot)$ is sub-additive. This gives property (3).  To prove property (4), we again use (\ref{eq:multi-derivative}) and observe
$$\Psi_{\calO_1}^{(J)}(X)\not = 0\ \Rightarrow\ \Phi_{\epsilon}^{(J_i)}(X_i) \not =0\ \text{\ for all $i\in [d]$}.$$
For any $i \in [d]$ such that $J_i\ne 0$, the latter implies that $0\le X_i\le \epsilon$
  since $\Phi_{\epsilon}(X_i)$ is constant outside $0 \le X_i \le \epsilon$.
This finishes the proof of Proposition \ref{prop:prod-molli}.
\end{proof}

\subsection{Lindeberg's replacement method and an $\Omega(n^{1/4 - c})$-query lower bound}
\label{sec:lindeberg}

Let $\bu_i$ and $\bv_i$, $i\in [n]$, denote independent random variables distributed according to
  $\bu$ and $\bv$ from Proposition \ref{prop:matchmoments} and \ref{prop:negsupport}
  with $\red{\ell}=\mp$ and $\mu=\mu(\mp)$, for some odd constant $\mp=\mp(c)\in \N$ to be specified at the end of this subsection.
  \red{We note that only in this subsection, Section \ref{sec:lindeberg}, do we take $\ell=h$ rather than $\ell=h^3$ (for the
  $\Omega(n^{1/4 - c})$ lower bound that we establish in this subsection, we only require $\ell=h$).}

Let ${\cal X}\in \{\pm 1/\sqrt{n}\}^{d\times n}$ denote a query matrix,
  and let $\calX^{(i)}$ denote its $i$th column.
Recall that
\begin{equation}\label{S-and-T}
\bS = \sum_{i=1}^n \bu_i \calX^{(i)} \quad\text{and}\quad \bT = \sum_{i=1}^n \bv_i \calX^{(i)}.
\end{equation}
Our goal is to show that
$\duo(\bS,\bT)\le 0.1$ when $d=O(n^{1/4-c})$.

To this end, let $\calO$ denote a union of orthants such that
\begin{equation}\label{eq:morning}
\duo(\bS,\bT)=\big|\hspace{-0.03cm}\Pr[\bS\in \calO]-\Pr[\bT\in \calO]\hspace{0.01cm}\big|.
\end{equation}
Following~\cite{Mos08,GOWZ10}, we first use the Lindeberg replacement
  method to bound $$\big|\hspace{-0.04cm}\E[\Psi_\calO(\bS)]-\E[\Psi_\calO(\bT)]\hspace{0.01cm}\big|,$$
  and then apply Proposition \ref{simplepro} to bound (\ref{eq:morning}).

For all $i\in \{0,1\ldots,n\}$ we introduce the $\R^d$-valued hybrid random variable:
\ignore{
}
\[
\bQ^{(i)} = \sum_{j=1}^{i} \bv_j \calX^{(j)} + \sum_{j=i+1}^n
\bu_j \calX^{(j)},\]
and note that $\bQ^{(0)} = \bS$ and $\bQ^{(n)} = \bT$. Informally we think of getting $\bT$ from $\bS$ via $\bQ^{(1)},\ldots,\bQ^{(n-1)}$ by swapping out each of the summands $\bu_j\calX^{(j)}$ for $\bv_j \calX^{(j)}$ one by one. The main idea is to bound the difference in expectations
\begin{equation} \label{eq:boundme}
\big|\hspace{-0.03cm}\E [\Psi_{\calO} (\bQ^{(i-1)} )]- \E[
\Psi_{\calO} (\bQ^{(i)} )] \hspace{0.01cm}\big|,
\end{equation}
since summing over all $i\in [n]$ gives an upper bound on
$$\big|\hspace{-0.03cm}\E [\Psi_{\calO} (\bS) ] - \E[\Psi_{\calO} (\bT)]\hspace{0.01cm}\big|
= \big|\hspace{-0.03cm}\E [\Psi_{\calO} (\bQ^{(0)} ) ] - \E[
\Psi_{\calO} (\bQ^{(n)} )] \hspace{0.01cm}\big|
\le \sum_{i=1}^{n} \big|\hspace{-0.03cm}\E [\Psi_{\calO} (\bQ^{(i-1)} )]- \E[
\Psi_{\calO} (\bQ^{(i)} )] \hspace{0.01cm}\big|
$$ via the triangle inequality.

To bound (\ref{eq:boundme}), we define the random variable
\begin{equation} \mathbf{R}_{-i} = \sum_{j=1}^{i-1} \bv_j \calX^{(j)} + \sum_{j=i+1}^n \bu_j \calX^{(j)} \label{eq:R-minus-i}
\end{equation}
and note that
\[ \big|\hspace{-0.03cm}\E[\Psi_{\calO} (\bQ^{(i-1)} ) ] - \E[\Psi_{\calO} (\bQ^{(i)} )]
\hspace{0.01cm}\big|=\big|\hspace{-0.03cm} \E[\Psi_{\calO} (\mathbf{R}_{-i} + \bv_{i} \calX^{(i)})] - \E[\Psi_{\calO} (\mathbf{R}_{-i} + \bu_i \calX^{(i)})] \hspace{0.01cm}\big|.
\]
Truncating the Taylor expansion of $\Psi_\calO$ at the $\mp$-th term
(Fact~\ref{fact:taylor}), we get
\begin{equation}\begin{aligned}
\hspace{-0.4cm}\E\hspace{-0.05cm}\big[\Psi_{\calO} (\mathbf{R}_{-i} + \bv_{i} \calX^{(i)})\big]
 &= \sum_{|J| \le \mp} \frac{1}{J!} \cdot
  \E\left[\Psi_{\calO}^{(J)} (\mathbf{R}_{-i}) \cdot  (\bv_{i} \calX^{(i)})^J
\right]
  \\
&\hspace{0.36cm}+ \hspace{-0.06cm}\sum_{|J|=\mp+1} \frac{\mp+1}{J!} \cdot
 \E\left[(1-\btau)^{\mp}\cdot \Psi_{\calO}^{(J)} (\mathbf{R}_{-i} + \btau \cdot \bv_{i} \calX^{(i)} ) \cdot (\bv_{i} \calX^{(i)})^{J}\right] \label{taylor-error}
\end{aligned}\end{equation}
where $\btau$ is a random variable uniformly distributed on the interval $[0,1]$ (so the very last expectation is with respect to $\btau$, $\bv_i$
and $\mathbf{R}_{-i}$).
Writing the analogous expression for $\E[\Psi_{\calO} (\bR_{-i} + \bu_{i} \calX^{(i)})]$,
\red{we observe that by Propositions \ref{prop:matchmoments} and~\ref{prop:negsupport} the first sums are equal term by term, i.e. we have
\[
\sum_{|J| \le \mp} \frac{1}{J!} \cdot
  \E\left[\Psi_{\calO}^{(J)} (\mathbf{R}_{-i}) \cdot  (\bv_{i} \calX^{(i)})^J
\right] = 
\sum_{|J| \le \mp} \frac{1}{J!} \cdot
  \E\left[\Psi_{\calO}^{(J)} (\mathbf{R}_{-i}) \cdot  (\bu_{i} \calX^{(i)})^J
\right]
\]
for each $|J| \leq h.$}  Thus we may cancel
all but the last terms to obtain
\[
\big|\hspace{-0.03cm}\E[\Psi_{\calO} (\bQ^{(i-1)} )]
- \E[\Psi_{\calO} (\bQ^{(i)} )]\hspace{0.01cm} \big| \le \sum_{|J|=\mp+1} \frac{\mp+1}{J!}\cdot \Vert \Psi_{\calO}^{(J)} \Vert_{\infty} \cdot \left( \E\big[|(\bv_i \calX^{(i)})^J|\big]+\E\big[|(\bu_i \calX^{(i)})^{J}|\big]\right).
\]
Observe that there are $|\{ J \in \N^d \colon |J| = \mp+1\}| =
\Theta(d^{\mp+1})$ many terms in this sum.
Recalling that each coordinate of $\calX^{(i)}$ has magnitude $1/\sqrt{n}$,
that both $\bu_i$ and $\bv_i$ are supported on at most $\mp \red{+1}$
real values that depend only on $\mp$ (by Propositions
\ref{prop:matchmoments} and \ref{prop:negsupport}),
and Proposition \ref{product-mollifier},
we have
\begin{equation}\label{taylor-error2}
\big|\hspace{-0.03cm}\E[\Psi_{\calO} (\bQ^{(i-1)} )] -
\E[\Psi_{\calO} (\bQ^{(i)} ) ]\hspace{0.01cm}\big| = O_\mp(1) \cdot  \left(\frac{d}{\epsilon}\right)^{\mp+1} \cdot  \frac{1}{n^{(\mp+1)/2}}.
\end{equation}
Summing over all $i\in [n]$ costs us a factor of $n$ and so we get
$$
\big|\hspace{-0.03cm}\E[\Psi_{\calO} (\bS)] -
\E[\Psi_{\calO} (\bT)]\hspace{0.01cm}\big| = O_\mp(1) \cdot  \left(\frac{d}{\epsilon}\right)^{\mp+1} \cdot  \frac{1}{n^{(\mp-1)/2}}.
$$

With this in hand we are in place to apply Proposition~\ref{simplepro}.
Let
$$\red{\calB}_\eps=\big\{X\in {\cal O}: \text{$|X_i|\le \epsilon$ for some $i\in [d]$}\hspace{0.01cm}\big\}.$$
Since both $\bv$ and $\bu$ are supported on values of magnitude $O_\mp(1)$,
we have that both $\Pr[ \bS \in \red{\calB}_\eps]$ and $\Pr[ \bT \in \red{\calB}_\eps]$
are bounded by $O_{\mp}(d \eps) + O_{\mp}(d/\sqrt{n})$
by using the standard $1$-dimensional Berry-Esseen
inequality (Theorem \ref{thm:be})
together with a union bound across the $d$ dimensions.
So all in all we have
\[
\duo(\bS,\bT) \leq O_{\mp}(d\hspace{0.02cm} \eps) + O_{\mp}(d/\sqrt{n}) + O_{\mp}(1)\cdot \left(\frac{d}{\epsilon}\right)^{\mp+1} \cdot  \frac{1}{n^{(\mp-1)/2}}. \]
We note as an aside at this point that
given any $0 < c < 1/4$, we may take $\eps = n^{-1/4}$ and take
$\mp$ \blue{to be the smallest odd integer at least $1/c$.
Then the RHS above is $O_{\mp}(n^{-c})$ when $d=O(n^{1/4-c})$}
as desired.  This gives the $\Omega(n^{1/4 - c})$ query lower bound claimed earlier:

\begin{proposition} \label{prop:warmup}
Given any $0 < c < 1/4$, there is a $\red{\kappa} = \red{\kappa}(c)>0$ such that
any non-adaptive algorithm for testing whether $f: \{-1,1\}^n \to \{-1,1\}$ is
monotone versus $\red{\kappa}$-far from monotone must use $\Omega(n^{1/4 - c})$ queries.
\end{proposition}

\subsection{Going beyond $\Omega(n^{1/4})$} \label{sec:goingbeyond}

\blue{The setup for the $\Omega(n^{1/2-c})$ bound
 is exactly the same as that of the $\Omega(n^{1/4-c})$ bound
  except~that $\bu_i,\bv_i$ are distributed according to $\bu$ and $\bv$
  from Proposition \ref{prop:matchmoments} and \ref{prop:negsupport}, respectively,
  with $\red{\ell}={\color{red}\mp^3}$ and $\mu=\mu(\red{\ell})$, for some odd constant
  $\mp=\mp(c)\in \N$ to be specified later
\red{(see Equation (\ref{eq:mp}))}.
We then repeat Lindeberg's replacement method on two random variables $\bS$ and $\bT$
  as defined in (\ref{S-and-T}),
  but only using the first $\mp$ matching moments of $\bu_i$ and $\bv_i$
  (with the higher $h^3-h$ matching moments being
  reserved for another application of Lindeberg's method later,
as mentioned in ``(4):  Handing pruned query sets''
in Section \ref{sec:techniques} above).}

The improvement to the $\Omega(n^{1/2-c})$ bound comes from a more careful analysis of the sum~in~(\ref{taylor-error}) which in turn translates into a stronger bound on the difference (\ref{eq:boundme}) than that was given in (\ref{taylor-error2}).
Specifically, rather than using the naive bound
\[ \big|\Psi_{\calO}^{(J)} (\mathbf{R}_{-i} + \btau \cdot \bv_{i} \calX^{(i)} )\big| \le \Vert \Psi^{(J)}_\calO \Vert_\infty = O_\mp(1)\cdot (1/\eps)^{\mp+1}\]
for each of the $\Theta(d^{\mp+1})$ possible outcomes of
$J\in \N^d$ (which shows up as the $\red{O_{\mp}(1)} \cdot (d/\eps)^{\mp+1}$ term in (\ref{taylor-error2})), we shall instead argue that almost all of these outcomes
actually make a much smaller contribution than $\red{O_{\mp}(1)} \cdot (1/\eps)^{\mp+1}$.
For this purpose, we will leverage the fourth property of $\Psi_\calO$ from Proposition~\ref{product-mollifier}; note that the proof of the $\Omega(n^{1/4-c})$ lower bound in Section \ref{sec:lindeberg} uses the first three properties of $\Psi_\calO$ from Proposition~\ref{product-mollifier}, but not the fourth.

Recall $\epsilon$ is the parameter of our mollifier $\Psi_\calO(\cdot)$.
Throughout the rest of the paper~we~shall take
\begin{equation}
\blue{\eps = n^{4/\mp-1/2} \quad\text{and} \quad
\delta = n^{-1/2}}
\label{eq:eps}
\end{equation}
but we continue to write ``$\eps$'' and ``$\delta$'' as separate parameters for conceptual clarity.
\blue{See Table \ref{table:params} as a reference for parameter settings used from
  Section \ref{sec:goingbeyond} through the rest of the paper.
   }

 \begin{table}[t]
\begin{center}

\begin{tabular}{l|l}

Parameter settings & Where the parameters are set\\

\hline\hline

$h=h(c)=$ smallest odd integer $> 5/c$ & Equation (\ref{eq:mp})\\
$\ell = h^3$ & Section \ref{sec:distributions}\\
$\mu=\mu(\ell)$ & Proposition \ref{prop:matchmoments}\\
$\eps = n^{4/\mp - 1/2}$ & Equation (\ref{eq:eps})\\
$\delta = n^{-1/2}$ & Equation (\ref{eq:eps})\\
$\beta = O_{\mp}(1)$ & Equation (\ref{eq:beta})

\end{tabular}
\end{center}
\caption{Parameter settings used from Section~\ref{sec:goingbeyond} onward.  The value ``$c$''
may be any positive\newline absolute constant.}
\label{table:params}
\end{table}

Revisiting equation (\ref{taylor-error}) of the proof above, we have that
\begin{eqnarray*}\begin{aligned}
&\hspace{-0.2cm}\big|\hspace{-0.03cm}\E[\Psi_{\calO} (\bQ^{(i-1)} ) ] - \E[\Psi_{\calO} (\bQ^{(i)} )]
\hspace{0.01cm}\big|\\[0.8ex]
&\hspace{-0.2cm}\le O_\mp(1) \sum_{|J| = \mp+1} 
\left(
 \E\hspace{-0.04cm}\left[\big| \Psi_{\calO}^{{(J)}} (\mathbf{R}_{-i} + \btau \cdot \bv_{i} \calX^{(i)} ) \cdot (\bv_{i} \calX^{(i)})^{J}\big|\right]+
 \E\hspace{-0.04cm}\left[\big| \Psi_{\calO}^{{(J)}} (\mathbf{R}_{-i} + \btau \cdot \bu_{i} \calX^{(i)} ) \cdot (\bu_{i} \calX^{(i)})^{J}\big|\right]\right)
\end{aligned}\end{eqnarray*}
For each multi-index $J$ with $|J|=\mp+1$ we relax
\begin{equation}
\E\hspace{-0.04cm}\left[\big| \Psi_{\calO}^{(J)} (\mathbf{R}_{-i} + \btau \cdot \bv_{i} \calX^{(i)} ) \cdot (\bv_{i} \calX^{(i)})^{J}\big|\right]\le
 \E\hspace{-0.04cm}\left[ \big|(\bv_{i} \calX^{(i)})^{J}\big| \cdot \sup_{T\in \red{[-\beta\delta,\beta\delta]^d}} \E\Big[ \big| \Psi_\calO^{(J)}
(\bR_{-i} + T)\big| \Big] \right], \label{eq:beta}
\end{equation}
where $\red{\beta}=O_{\mp}(1)$ is an absolute constant that depends
  only on the largest value in the support~of~$\bv$
(which depends only on $\mp$).  Observe that since each coordinate of
$\calX^{(i)}$ has magnitude $1/\sqrt{n}$, each coordinate of the vector-valued random variable
$\btau \cdot \bv_i\calX^{(i)}$ is supported on values in $\red{[-\beta\delta,\beta\delta]}$, for
  the $\red{\beta}$ as described above.
Combining the above with an analogous bound for the $\bu_i$ term, we have
\begin{equation}\label{eq:lindeberg}
\big|\hspace{-0.03cm}\E[\Psi_{\calO} (\bQ^{(i-1)} ) ] - \E[\Psi_{\calO} (\bQ^{(i)} )]
\hspace{0.01cm}\big|
\le \frac{O_\mp(1)}{n^{(\mp+1)/2}}  \sum_{|J| = \mp+1}\hspace{-0.08cm}\left(
\hspace{0.05cm}\sup_{T\in \red{[-\beta\delta,
\beta\delta]^d}} \E\hspace{-0.04cm}\Big[ \big| \Psi_\calO^{(J)}(\bR_{-i} + T)\big| \Big]\right).
\end{equation}

We obtain an improved upper bound on this sum by exploiting the distributional
properties of the $d$-dimensional
random variable $\mathbf{R}_{-i}+T$. In particular we
would like to show that for~most~ways of choosing
$\mp+1$ out of the $d$ coordinates, it is quite unlikely that all $\mp+1$ chosen coordinates~can
simultaneously take a value in the small interval $\red{[-\beta\delta,\blue{\beta}\delta]}.$ (Note that almost all $J$ with $|J| = \mp+1$ satisfy $\#J = \mp+1$.) The fourth property of $\Psi_\calO$ from Proposition~\ref{product-mollifier} implies that having all these coordinates be small is the only way an outcome
of $\mathbf{R}_{-i} + T$ can have
$$\E\hspace{-0.04cm}\Big[ \big| \Psi_\calO^{(J)}(\bR_{-i} + T)\big| \Big]$$
make a nonzero contribution to the sum in (\ref{eq:lindeberg}). In other words, we would like
to use the fact that for all $J \in \N^d$ with $|J|=\mp+1$ we have
\begin{equation}
\label{eq:useme}
\sup_{T\in \red{[-\beta\delta,{\beta}\delta]^d}} \E\hspace{-0.04cm}\Big[ \big| \Psi_\calO^{(J)}
(\bR_{-i} + T)\big| \Big] \le
O_\mp(1)\cdot \left(\frac1{\eps}\right)^{\mp+1}\cdot \Pr\big[(\bR_{-i})|_J \in \red{
\calB_{J}} \big],
\end{equation}
where we use $\red{\calB_{J}}$ to denote the origin-centered $(\#J)$-dimensional
box 
  $\red{[\hspace{0.02cm}-{\epsilon-\beta\delta},{\epsilon+\beta\delta}\hspace{0.04cm}]^{\#J}}$.
Recall that the analysis of the previous subsection simply used
the weaker bound obtained from (\ref{eq:useme}) by upper bounding
$\Pr\hspace{0.02cm}[(\bR_{-i})|_J \in \red{\calB_{J}}]$ by $1.$

Unfortunately, given an arbitrary
query set, we cannot argue that the RHS of (\ref{eq:useme})
is typically small.  Indeed, consider a $d$-query set $\calX$
in which a single fixed string $Q \in \{\pm 1/\sqrt{n}\}^n$
is repeated $d$ times.  In such a situation, \emph{every}
outcome of $J$ will have
$$\Pr\big[(\bR_{-i})|_J \in \red{\calB_{J}} \big]=
\Pr\big[(\bR_{-i})_1 \in \red{[\hspace{0.02cm}-{\epsilon-\beta\delta},{\epsilon+\beta\delta}\hspace{0.04cm}]}
\hspace{0.02cm} \big],$$
because every coordinate of every outcome of $(\bR_{-i})$ is the same,
and this probability over the 1-dimensional random variable
$(\bR_{-i})_1$ may be as large as $\Omega(\eps)$;
thus no significant savings is achieved over the earlier analysis.
However, it is clear that such a query set $\calX$ is highly
``degenerate,'' in~the sense that it can be replaced by a 1-query
set (which we denote by $\calX^\ast$) consisting of just one copy of $Q$,
which will serve just as well as $\calX$ for the purpose of monotonicity testing.
(More precisely, the ``union-of-orthants'' distance $\duo(\bS,\bT)$
corresponding to the original query set will be precisely the same
as the union-of-orthants distance $\duo(\bS^\ast,\bT^\ast)$
corresponding to the reduced query set $\calX^\ast$.)

Is it possible that \emph{every} ``degenerate'' query set
(for which (\ref{eq:lindeberg}) is large) can be  ``pruned''
down to an essentially equivalent query set (in terms of our $\duo$
measure) for which we can give a strong upper bound?  Perhaps
surprisingly, the answer is yes; however, doing this requires significant
work and careful analysis.  In the next section
we describe and analyze our pruning procedure, and \orange{in Section~\ref{sec:lbscattered}} we show how an analysis based on (\ref{eq:useme}) can
handle pruned query sets.

\section{Pruning a query set}\label{sec:pruning}

In this section we explain how an arbitrary query set can
be ``pruned'' so as to make it ``scattered.''  (The definition of a ``scattered''
query set is somewhat complicated, involving the density of points that lie close to the
linear span of other sets of points, so we defer it to Section \ref{sec:prune}.)  We show that the pruning procedure has only a negligible effect on the variation distance
$\duo(\bS,\bT)$ that we are aiming to bound.  In later sections we give a lower bound against scattered query sets and thereby prove our main result.

We give some preliminary geometric results
  in Section \ref{sec:odlyzko}, and after some setup in Section
\ref{sec:compatibility}, describe and analyze the pruning procedure in Section
\ref{sec:prune}.

\subsection{Useful results about hypercubes and subspaces} \label{sec:odlyzko}

The first geometric result we require is a variant of a well known fact due to Odlyzko~\cite{Odl88}.
We begin by recalling the original fact:

\begin{fact} \label{fact:o}
Let ${\calV} \subseteq \R^n$ be a subspace of dimension $k$.  Then
$|\hspace{0.02cm}{\calV} \cap \{\pm 1/\sqrt{n}\}^n| \leq 2^k.$
\end{fact}

Our variant is more restrictive than the original statement
in that it only deals with subspaces ${\calV}$ of the
form ${\calV}=\span\{{V^{(1)}},\dots,{V^{(k)}}\}$, for some ${V^{(1)},\dots,V^{(k)}} \in \{\pm 1/\sqrt{n}\}^n$ (though see Remark
\ref{rem:gen-od}).
However, the variant is significantly more general in that it gives us a bound
on the number of Hamming balls that are required to cover all points
of $\{\pm 1/\sqrt{n}\}^n$ that lie \emph{close} to (and need not lie exactly on)
the subspace ${\calV}$.  (Odlyzko's fact may be viewed as giving a bound on the number
of radius-0 Hamming balls that are required to cover all points
of $\{\pm 1/\sqrt{n}\}^n$ that lie exactly on ${\calV}$.)
A detailed statement and proof of our variant follow.

Given ${r} \geq 0$ and a subspace
${\calV} \subseteq \R^n$, we define the
\emph{${r}$-dilation of ${\calV}$} to be the set
\[
B_{\ell_2}({\calV},{r}) := \bigcup_{{V \in \calV}} B_{\ell_2}({V},{r}).
\]

Our lemma is the following:
\newcommand{\cover}{\mathrm{cover}}

\begin{lemma}
\label{lem:fo}
Given any set ${\calA = \{V^{(1)},\dots,V^{(k)}\}} \sse  \{\pm1/\sqrt{n}\}^n$
  and any $r \geq 0$, there exists
a set of at most {$2^{k^2}$} points $\cover({\calA})  \sse  \{\pm 1/\sqrt{n}\}^n$ such that
\[
B_{\ell_2}(\span({\calA}),r) \cap \{\pm 1/\sqrt{n}\}^n
\subseteq \bigcup_{{Y\in\cover(\calA)}} B_{\mathrm{Ham}}({Y},r^2n).
\]
\end{lemma}

Observe that by taking $r = 0$,  Lemma \ref{lem:fo} recovers Fact \ref{fact:o}
for $\calV=\span \{\hspace{0.01cm} V^{(1)},\dots,V^{(k)} \hspace{0.01cm}\}$ where
$V^{(1)},\dots,V^{(k)}\in \{\pm 1/\sqrt{n}\}^n$, with the somewhat weaker
bound $2^{k^2}$ compared to $2^k$.

\begin{proof}
Fix any $\green{V} \in \{\pm 1/\sqrt{n}\}^n$ such that
$\green{V} \in B_{\ell_2}(\span(\calA), r)$, so there exists a  
  $\green{U} = \sum_{j=1}^k \alpha_j V^{(j)}$ such that
$\|\green{U-V}\|_2 \leq  r.$ Let
the vector $\green{U}_{\mathrm{round}} \in \{\pm 1/\sqrt{n}\}^n$ be defined by taking
$$(\green{U}_{\mathrm{round}})_i = \sign(\green{U_i})\big/\sqrt{n} \in \{\pm 1/\sqrt{n}\}$$ for each $i \in [n]$.
It is clear that we have
\begin{eqnarray*}
\Vert \green{U}_{\mathrm{round}} - \green{V} \Vert_2 = \frac{2}{\sqrt{n}} \sqrt{\sum_{i=1}^n \mathbf{1}\big[
\green{V_i} \not = (\green{U}_{\mathrm{round}})_i\big]} \le 2 \cdot \sqrt{\sum_{i=1}^n (\green{U_i-V_i})^2}
= 2\cdot \Vert \green{U-V}\Vert_2 \le 2r,
\end{eqnarray*}
and also that
\[ \| \green{U}_{\mathrm{round}} -\green{V}\|_2 = \sqrt{\frac{4\cdot  d_{\mathrm{Ham}}
(\green{U}_{\mathrm{round}},\green{V})}{n}}. \]
As a result, we have\hspace{0.02cm}
$d_{\mathrm{Ham}}(\green{U}_{\mathrm{round}},\green{V}) \leq r^2 n.$

Let $\cover(\calA)\subseteq \{\pm 1/\sqrt{n}\}^n$ denote the following set of points:
$$\cover(\calA) = {\big\{\green{U}_{\mathrm{round}}: \green{U} \in \span(\calA)
\big\}}.$$
We will show that $|\hspace{0.015cm}\cover(\calA)| \le 2^{k^2}$; this establishes the lemma.
To see this, note that
$$
(\green{U}_{\mathrm{round}})_i = \sign \Bigg(\sum_{j=1}^k \alpha_j \cdot V^{(j)}_{i} \Bigg),\quad\
\text{{given $\green{U}=\sum_{j=1}^k \alpha_j\cdot V^{(j)}$.}}
$$
In other words, the $i$-th entry of $\green{U}_{\mathrm{round}}$ is given by the value of the
$k$-variable LTF
\begin{eqnarray*}
&{f(Y) = \sign \Big(\sum_{j=1}^k \alpha_j {Y_j}\Big)}&
\end{eqnarray*}
evaluated on the fixed input $X^{(i)} = (V^{(1)}_i,\dots,V^{(k)}_i) \in \{\pm 1/\sqrt{n}\}^k$
(note that different $\green{U}_{\mathrm{round}}$'s~correspond to LTFs
  with different coefficients, but the $n$ inputs $X^{(1)},\dots,X^{(n)}$ on which the LTFs
are evaluated are the same over all $\green{U}_{\mathrm{round}}$'s).  
Thus we can upper bound the number of distinct vectors $\green{U}_{\mathrm{round}}$ by the number of distinct $k$-variable LTFs (viewed as Boolean functions)
over $\{\pm 1/\sqrt{n}\}^{k}$, which is at most $2^{k^2}$ by Fact \ref{fact:number-halfspaces}.
\end{proof}

\begin{remark}
\label{rem:gen-od}
Though we do not need it, we note that Lemma \ref{lem:fo}
may easily be generalized to allow each of $V^{(1)},\dots,V^{(k)}$ to be an arbitrary
point in $\R^n$, at the cost of having the RHS become $n^{k+1}$
instead of $2^{k^2}$. As the VC dimension of the class of all LTFs over $\R^k$
is $k+1$, Sauer's lemma tells us that the number of different ways that LTFs can label
a fixed set of $n$ points in $\R^k$ (like the points $X^{(1)},\dots,X^{(n)}$) is at most
$(en/(k+1))^{k+1} \leq n^{k+1}$.
\end{remark}

The next geometric lemma that we require is the following:

\begin{lemma}\label{lem:low}
Fix any positive integer $\mp$.
\hspace{-0.05cm}There exist two constants ${\gamma}_1={\gamma}_1(\mp)$ and
  ${\gamma}_2={\gamma}_2(\mp)$ with the following property.
For any $\calA=\{V^{(1)},\ldots,$ $V^{(k)}\}\subset \{\pm 1/\sqrt{n}\}^n$ with $k\le \mp$
and any $V\in$ $ \{\pm 1/\sqrt{n}\}^n$,
there is a vector
%
$U = \beta_1 V^{(1)} + \cdots + \beta_k V^{(k)} \in \span(\calA)$
such that $|\beta_i|\le \gamma_1 \text{~for all~}i$ and
$$\|V-U\|_2\le \gamma_2\cdot d_{\ell_2} \big(V,\hspace{0.03cm}\span(\calA) \big).
$$
\end{lemma}

{Roughly speaking, Lemma~\ref{lem:low} shows that given any set of
  $k\le \mp$ vectors $\calA = \{V^{(1)},\ldots,V^{(k)}\}$ from $\{\pm 1/\sqrt{n}\}^n$ and a ``target vector'' $V \in \{\pm 1/\sqrt{n}\}^n$, there exists $U\in
{\span(\calA)}$ such that $U$ is
almost as close to $V$ in Euclidean distance as the closest point in
$\span(\calA)$, and $U$ can be written as a ``low-weight'' linear combination of the elements in $\calA$.} Note that there are competing demands imposed by keeping both parameters ${\gamma}_1$ and ${\gamma}_2$ small;\ignore{\orange{(i.e.~the distance between $U$ and $V$, and the weight of the linear combination)};} for example, it is easy to see that either one may individually be made to be 1, but doing this may potentially cause the other one to become large.  The crux of Lemma \ref{lem:low} is that it is possible to simultaneously have both $\red{\gamma}_1$ and $\red{\gamma}_2$ bounded by $O_\mp(1)$
independent of $n$.

\begin{proof}
\ignore{
We start by defining a constant $K=K(\mp)=O_\mp(1)$.
Let $K$ denote the smallest real number such that
  for every $k\le \mp$ and every $k\times k$ nonsingular matrix $M\in \{\pm 1\}^{k\times k}$,
  the largest absolute value of entries of $M^{-1}$ is at most $K$.
It is clear that $K$ only depends on $\mp$.
}
Given $\calA$ and $V$, we let $U=\beta_1 V^{(1)}+\cdots+\beta_k V^{(k)}$ denote the
  closest point to $V$ in $\span(\calA)$. 
Below we view~$\calA$ as a $k\times n$ matrix, with $V^{(i)}$
being its $i$-th row vector.
Note that $\calA$ has ${m} \le 2^k$
many distinct columns, and
  we let {$P^{(1)},\ldots,P^{(m)}$} denote these column vectors in $\{\pm 1/\sqrt{n}\}^k$.
Let $I\subseteq [{m}]$ denote the set of indices $i\in [m]$ such that coordinates of $U$
  that correspond to columns of type $P^{(i)}$ have absolute value at most $2/\sqrt{n}$.
  (Note that if two coordinates $U_a,U_b$ of $U$ correspond to the same column type $P^{(i)}$
  then $U_a=U_b$.)

We consider two cases.
For Case 1, we show that $\beta_1,\ldots,\beta_k$ already satisfy $|\beta_i|=O_{\mp}(1)$ for all $i\in [k]$,
  and we are done.
For Case 2, we use $\beta_1,\ldots,\beta_k$ to obtain $\alpha_1,\ldots,\alpha_k$ such that
  $|\alpha_i|=O_{\mp}(1)$ for all $i\in [k]$ and $W=\alpha_1V^{(1)}+\cdots+\alpha_kV^{(k)}$ has small Euclidean distance
  from $V$ as claimed.

\textbf{Case 1}: The set of columns in $\{P^{(i)}:i\in I\}$ spans full dimension $k$.
For this case we pick any $k$ such columns, say $P^{(1)},\ldots,P^{(k)}$ without loss of generality,
  in $I$.
Then $(\beta_1,\ldots,\beta_k)$ is the unique solution to the following linear system
  of $k$ equations in variables $x_1,\ldots,x_k$:
$$
P^{(i)}\cdot (x_1,\ldots,x_k)=P^{(i)}\cdot (\beta_1,\ldots,\beta_k),\ \ \ \text{for $i\in [k]$.}
$$
Each entry of the $k \times k$ coefficient matrix given by the $P^{(i)}$'s is $\pm 1/\sqrt{n}$,
and the right side of each of the $k$ equations has absolute value at most $2/\sqrt{n}$.
By Cramer's rule it follows that $|\beta_i|=O_k(1)=O_\mp(1)$ for all $i$, and the lemma is proved in this case.

\textbf{Case 2}: The set of columns in $\{P^{(i)}:i\in I\}$ spans a space of dimension $\red{j}<k$.
For this case we pick $\red{j}$ independent columns from $\{P^{(i)}:i\in I\}$, say $P^{(1)},\ldots,P^{(\red{j})}$.
Then we pick arbitrarily~$k-\red{j}$ vectors $T^{({\red{j}+1})},\ldots,T^{(k)}$ from $\{\pm 1/\sqrt{n}\}^n$
  so that they together with $P^{(1)},\ldots,P^{(\red{j})}$ span full dimension $k$
  (note that $T^{(i)}$'s are not necessarily column vectors of $\calA$).
Solving the following linear system we get an alternative set of
  coefficients $\alpha_1,\ldots,\alpha_k$:\vspace{0.1cm}
\begin{enumerate}
\item For each $i\in [\red{j}]$, we require
  $P^{(i)}\cdot (x_1,\ldots,x_k)=P^{(i)}\cdot (\beta_1,\ldots,\beta_k)\in [-2/\sqrt{n},2/\sqrt{n}]$.\vspace{-0.08cm}
\item For each $i\in [\red{j}+1:k]$, we require ${T}^{(i)}\cdot (x_1,\ldots,x_k)=0$.\vspace{0.1cm}
\end{enumerate}
Let $(\alpha_1,\ldots,\alpha_k)$ denote the unique solution to this linear system.
Similar to Case 1, Cramer's rule implies that $|\alpha_i|=O_k(1) = O_\mp(1)$ for all $i\in [k]$.

Finally we complete the proof by showing that the vector
$W=\alpha_1V^{(1)}+\cdots+\alpha_kV^{(k)}$ is close to $V$; more precisely, we show that
$$
\|V-W\|_2=O_\mp(1)\cdot \|V-U\|_2
=O_\mp(1)\cdot d_{\ell_2}\big(V,\hspace{0.03cm}\span(\calA)\big).
$$
For this we just compare $W=\alpha_1V^{(1)}+\cdots+\alpha_kV^{(k)}$
  with $U=\beta_1V^{(1)}+\cdots+\beta_kV^{(k)}$
  entry by entry.
Fix any $a\in [n]$ and suppose that the $a$-th column of $\calA$ is of type
  ${P}^{(b)}$, for some
  $b\in [{m}]$.
If $b\in I$, then it is clear that $U_a=W_a$.
If $b\notin I$, then we have $|U_a-V_a|> 1/\sqrt{n}$ since
  $|U_a|>2/\sqrt{n}$.
On the other hand, from $|\alpha_i|=O_\mp(1)$ for all $i$ we also have
  $|W_a|\le k\cdot O_\mp(1)/\sqrt{n}$ and thus,
$$
|W_a-V_a|=O_\mp(1)/\sqrt{n}<O_\mp(1)\cdot |U_a-V_a|.
$$
The claim now follows.
\end{proof}

\subsection{Setup for the pruning procedure: compatibility between points and sets}
\label{sec:compatibility}
We will use the following simple lemma, which follows directly from the Hoeffding inequality
and the fact that $\bu,\bv$ are bounded and $\E[\bu]=\E[\bv]$.
Recall that $\mp=\mp(c)$ is an odd integer constant.

\def\uu{\bu} \def\ww{\bw}
\def\vv{\bv}

\begin{lemma}\label{lem:hoeffding}
Let $\ww_1,\ldots,\ww_n$ denote $n$ independent random variables, where
  each $\ww_i$ is distributed according to either $\uu$ or $\vv$ given in
  Proposition \ref{prop:negsupport} or \ref{prop:matchmoments} with \red{$\ell=\mp^3$
  and $\mu=\mu(\ell)$}.
Let $W\in \R^n$ and ${\bx}=\sum_{i\in [n]} \ww_i W_i$, then
$\E[{\bx}]=\red{\mu}\cdot \sum_{i\in [n]} W_i.$
Moreover, we have
$$
\Pr\left[\hspace{0.03cm}|\hspace{0.03cm}{\bx}-\E[\bx]\hspace{0.01cm}|
\ge \|W\|_2\cdot \red{(\log n)^{3/4}} \hspace{0.05cm}\right]\le \frac{1}{n^{\omega(1)}}.
$$
\end{lemma}


Now we define compatibility between
  a point $V\in \{\pm 1/\sqrt{n}\}^n$ and a set $\calA\subset\{\pm 1/\sqrt{n}\}^n$. Let~$\red{\gamma}_1$ $=\red{\gamma}_1(\mp)$ and $\red{\gamma}_2=\red{\gamma}_2(\mp)$ denote the constants from Lemma~\ref{lem:low}.
\blue{Recall that $\epsilon=n^{4/\mp-1/2}$.}
\begin{definition}[Compatibility] \label{def:incompatible}
Given $\calA=\{V^{(1)},\dots,V^{(k)}\}\subset \{\pm1/\sqrt{n}\}^n$
  for some $k\le \mp$ and $V\in$ $\{\pm 1/\sqrt{n}\}^n$,
we say that $V$ is \emph{incompatible} with $\calA$ if
  there \emph{exist} real numbers $\beta_1,\ldots,\beta_k$ such that both (i) $|\beta_i|\le \red{\gamma}_1(\mp)$
  for all $i\in [k]$ and (ii) the vector $U=\beta_1 V^{(1)}+\cdots +\beta_k V^{(k)}\in \span(\calA)$ satisfies
\begin{eqnarray*}
&\Big|\sum_{i\in [n]} (V_i-U_i)\hspace{0.02cm}\Big|> \Big(\|V-U\|_2
  +\epsilon\Big)\cdot \log n.&
\end{eqnarray*}
Otherwise we say $V$ is \emph{compatible} with $\calA$.
\end{definition}

We may equivalently define compatibility as follows: $V$ is compatible with $\calA$ if for \emph{every}
$\beta_1,\dots$ $\beta_k$ of magnitude at most $\red{\gamma}_1(\mp)$, the vector $U=\beta_1 V^{(1)}+\cdots +\beta_k V^{(k)}$
satisfies
\begin{eqnarray*}
&\Big|\sum_{i\in [n]} (V_i-U_i)\hspace{0.02cm}\Big| \leq \Big(\|V-U\|_2
  +\epsilon\Big)\cdot \log n.&
\end{eqnarray*}

Recall from (\ref{eq:useme})
that we would like to give a strong upper bound on $ \Pr[(\bR_{-i})|_J \in {\red{\calB}_{J}}]$ for as many multi-indices $J$ with $|J|=\mp+1$ as possible.
Given a fixed set $\calX$ of $d$ query strings, a subset $\calA \subset \calX\subset \{\pm 1/\sqrt{n}\}^n$ of size $k \leq \mp$ corresponds naturally to a multi-index $J$ with $|J|=|\calA|.$   It is intuitively helpful to think of
a multi-index $J$ as being ``built up'' by successively adding~elements
from $\calX$ to $\calA$ one by one, starting with $\emptyset$.  This motivates the above definition of
incompatibility; as the following lemma shows, if a query string $V$
is incompatible with $\calA$, then we get a very strong bound on the
probability $\Pr[(\bR_{-i})|_J \in \blue{\calB_{J}}]$ for the multi-index $J$
corresponding to $\blue{\{V\}} \cup \calA$ (which is desirable for our analysis).
We will use this lemma later in Section \ref{sec:propneed}
to deal with multi-indices corresponding to subsets of queries that contain a
query that is incompatible with the other queries.

\begin{lemma}\label{lem:incomp}
Suppose $V\in \{\pm 1/\sqrt{n}\}^n$ is incompatible with
  set $\calA \subset \{\pm 1/\sqrt{n}\}^n$ where $k=|\calA|\le\mp$.
Let $(\calA,V)$ be the $(k+1) \times n$ matrix whose rows are given by the
vectors of $\calA$ followed by $V$. Then
$$
\Pr\Big[(\calA,V)\cdot (\ww_1,\ldots,\ww_n)\in [-\blue{2\epsilon},\blue{2\epsilon}]^{k+1}\Big]=\frac{1}{n^{\omega(1)}},
$$
where $\ww_i$'s are independent random variables
  each of which is distributed according to $\uu$ or~$\vv$.
  \end{lemma}
\begin{proof}
Let $U \in \R^n$ be a linear combination of the elements of $\calA$ that satisfies conditions (i) and
(ii) of Definition \ref{def:incompatible} (the existence of $U$ is guaranteed by the incompatibility
of $V$ with $\calA$).
Because $U$ is a ``low-weight'' linear combination of  vectors in $\calA$,
having $(\calA,V)\cdot (\ww_1,\ldots,\ww_n)\in \blue{[-2\epsilon,2\epsilon]}^{k+1}$ implies that
  $W:=V-U$ also satisfies
\begin{equation}\big|\hspace{0.02cm}W\cdot (\ww_1,\ldots,\ww_n)\hspace{0.01cm}\big|
=O_\mp(\epsilon).\label{eq:wantthis}
\end{equation}

Next observe that by condition (ii) of Definition \ref{def:incompatible},
we have that
\begin{eqnarray*}
&\Big|\sum_{i \in [n]} W_i\hspace{0.03cm} \Big| > \big(\|W\|_2 + \eps\big) \log n.&
\end{eqnarray*}
On the other hand, Lemma \ref{lem:hoeffding} gives us that
\begin{eqnarray*}
&\Pr\Bigg[\hspace{0.05cm}\Big|\hspace{0.03cm}W \cdot (\bw_1,\dots,\bw_n) - \mu \sum_{i\in [n]} W_i\hspace{0.03cm}\Big|
\ge \|W\|_2\cdot (\log n)^{3/4} \hspace{0.05cm}\Bigg]\le&\hspace{-0.22cm} \frac{1}{n^{\omega(1)}};
\end{eqnarray*}
together with the previous inequality,
recalling that $0 < \mu=O_{\mp}(1)$,
this gives
\[
\Pr\Big[\hspace{0.02cm}\big|\hspace{0.03cm}W \cdot (\bw_1,\dots,\bw_n) \hspace{0.03cm}\big|
= \Omega\big((\|W\|_2 + \eps) \log n\big)\hspace{0.02cm}\Big] \geq 1 -  \frac{1}{n^{\omega(1)}}.
\]
This then implies that $|\hspace{0.01cm}W \cdot (\bw_1,\dots,\bw_n)\hspace{0.01cm}| = O_\mp(\eps)$
with probability at most ${1}/{n^{\omega(1)}}$, which together with (\ref{eq:wantthis})
establishes the lemma. 
\end{proof}

Finally, the following lemma plays a key role in arguing about our pruning procedure:

\begin{lemma}
\label{lem:partition}
Let $\calA = \{V^{(1)},\dots,V^{(k)}\} \subset  \{\pm1/\sqrt{n}\}^n$ where
  $k\le \mp$, and $r \geq 0$.
Let $\calR\subset\{\pm1/\sqrt{n}\}^n$ denote a set of points such that $\calR\cap \calA=\emptyset$
  and $\calR\subset B_{\ell_2}(\span(\calA),r)$.
Then one can partition the set $\calR$ into three disjoint sets
$\calR=\calR_{cover}\cup \calR_{remove}\cup \calR_{incomp}$
with the following properties:
\begin{flushleft}\begin{enumerate}
\item $\calR_{incomp}$ consists of all the points in $\calR$ that are incompatible with $\calA$;\vspace{-0.1cm}
\item $|\calR_{cover}|\le 2^{\mp^2}$; and\vspace{-0.1cm}
\item For each point $W\in \calR_{remove}$, there exists \blue{at least one}
  point \red{$V\in \calR_{cover}$}  such that $\|V-W\|_2\le 4r$. \blue{Moreover, every
  such $V\in \calR_{\cover}$ satisfies}
\begin{eqnarray}\label{eq:finaleq}
&\Big|\hspace{-0.03cm}\sum_i (V_i- W_i)\Big|\le \red{(r+\epsilon)}\log^2 n.&
\end{eqnarray}
\end{enumerate}\end{flushleft}
\end{lemma}

As their names suggest, the points in $\calR_{cover}$ will be used as a
``cover'' of the points in $\calR_{remove}$, which will be removed from the
query set in the pruning
procedure described later.
\blue{Also note that by condition (3), we must have $\calR_{remove}=\emptyset$
  when $r=0$.}

\begin{proof}
Let $\calR_{incomp}$ be as described in (1) above, and let $\calR'=\calR\setminus \calR_{incomp}$.

From Lemma \ref{lem:fo}, we know that there is a set
$\blue{\cover(\calA)} \subset \{\pm 1/\sqrt{n}\}^n$ such that
$|\blue{\cover(\calA)}| \le 2^{\mp^2}$ and for any
$W \in \blue{\calR'}$, there is a $V \in \blue{\cover(\calA)}$ such that  $\| V- W \|_2 \le 2r$.
It follows that there exists a set $\calR_{cover} \subseteq \calR'$
\blue{with $|\calR_{cover}| \leq |\cover(\calA)|\le  2^{\mp^2}$}
such that
for any $W \in \calR'$, there is a~$V \in \calR_{cover}$ such  that $\|V- W \|_2 \le 4r$. Let $\calR_{remove} = \calR' \setminus \calR_{cover}$. Then the only requirement that
remains to be proven is the second inequality in (\ref{eq:finaleq}).

To prove this, we let $\calA=\{V^{(1)},\ldots,V^{(k)}\}$ and
  let $U=\beta_1 V^{(1)}+\cdots+\beta_k V^{(k)}$ denote the vector guaranteed by
  Lemma \ref{lem:low} for $\calA$ and $V$, with $\|\beta_i\|\le \red{\gamma}_1(\mp)$ for all $i$
  and $$\|U-V\|_2\le \red{\gamma}_2(\mp)\cdot d_{\ell_2}(\span(\calA),V)\le
\ignore{\orange{5}} \red{\gamma}_2(\mp)\cdot r.$$
\ignore{\rnote{Added a ``$\orange{5}$'' in front of the $r$; we know that
$\|V-W\|_2 \leq 4r$ and $\|W - \span(\calA)\|_2 \leq r$ so this gives
$\|\span(\calA)-V\|_2 \leq r.$ \green{Xi: I think it is fine not to have $5$ since
  $V\in \calR_{cover}\subseteq \calR\subset B_{\ell_2}(\span(\calA),r)$.}}}
Note that $\beta_1,\ldots,\beta_k$ satisfy condition (i) of Definition \ref{def:incompatible}.
As $V$ is compatible with $\calA$, we have
\begin{eqnarray*}
&\Big|\sum_i (V_i-U_i)\hspace{0.02cm}\Big|\le \Big(
\|V-U\|_2+\epsilon\Big)\cdot \log n.&
\end{eqnarray*}
Similarly, as $W$ is compatible with $\calA$ as well, we have
\begin{eqnarray*}
&\Big|\sum_i( W_i-U_i)\hspace{0.02cm}\Big|\le \Big(
\|W-U\|_2+\epsilon\Big)\cdot \log n.&
\end{eqnarray*}
Combining these two inequalities, we have
$$\begin{aligned}
\left|\hspace{0.02cm}\sum (V_i-W_i)\hspace{0.02cm}\right|&\le \left|\hspace{0.02cm}\sum (V_i-U_i)\hspace{0.02cm}\right| +\left|\hspace{0.02cm}\sum (W_i-U_i)\hspace{0.02cm}\right|
\le \Big(
\|V-U\|_2+\|W-U\|_2+2\hspace{0.02cm}\epsilon\Big)\cdot \log n. 
\end{aligned}$$
Combining this with $\|V-U\|_2\le \ignore{\orange{5}}\red{\gamma}_2(\mp)\cdot r$ and
$$
\|W-U\|_2\le \|W-V\|_2+\|V-U\|_2 \le  \big(4+\ignore{\orange{5}}\blue{\gamma}_2(\mp)\big) \cdot r,
$$
the second part of (\ref{eq:finaleq}) is proven.
This finishes the proof of the lemma.
\end{proof}

\subsection{The pruning procedure and its analysis} \label{sec:prune}

Let $\calX=\{X^{(1)},\ldots,X^{(d)}\}\subseteq \{\pm 1/\sqrt{n}\}^n$ denote a query set of size $d$.
We view $\calX$ as a $d\times n$ matrix with $X^{(i)} \in \{\pm 1/\sqrt{n}\}^n$
being its $i$-th row vector and $\calX^{(j)}\in \{\pm 1/\sqrt{n}\}^{\red{d}}$ its $j$-th column vector.

Fix a $c>0$, we now specify the function $\mp$:
\begin{equation} \label{eq:mp}
\mp(c) = \text{\red{the smallest odd integer $\ge 5/c$,}}
\end{equation}
and recall that $\ell = \mp^3$.
Recall our goal is to show that any query set $\calX$ of size $d\le n^{1/2-c}$ satisfies
$$
\duo(\bS,\bT)=\max\Big\{ \big|\hspace{-0.02cm}\Pr[\bS \in \calO]-\Pr[\bT\in\calO]
\hspace{0.01cm}\big|
\colon \text{$\calO$ is a union of orthants in $\R^d$\Big\}}\le 0.1.
$$
Here $\bS=\sum_j \bu_j \calX^{(j)}$ and $\bT=\sum_j \bv_j \calX^{(j)}$,
  where $\bu_j$ and $\bv_j$ are independent random variables with the
  same distribution as $\uu$ and $\vv$ from Proposition \ref{prop:negsupport} and
  \ref{prop:matchmoments}, given constants $\red{\ell}$ and $\mu(\red{\ell})$.

Next we describe a procedure that ``prunes'' $\calX$
  and outputs a new query~set $\calX^*\subseteq \calX$,
  which is almost as good as ${\calX}$ for monotonicity testing, and is what we call a \emph{scattered} query set.


\ignore{
%
}

\begin{definition}[Scattered query sets] \label{def:scattered}
Fix $\calA \subseteq \calX$ with $\red{0 <} |\calA| \leq \mp$ and a value $r>0$. Let
$$\calR=\big(\calX\cap B_{\ell_2}(\span(\calA),r)\big)\setminus \calA,$$ and let
$\calR=\calR_{cover}\cup \calR_{remove}\cup \calR_{incomp}$
denote the partition of $\calR$ promised by Lemma \ref{lem:partition}.
We say that $\calA$ is \emph{$r$-scattered} if $\calR_{remove}$ satisfies
\begin{equation}\label{eq:scattered}
|\hspace{0.02cm}\calR_{remove}\hspace{0.02cm}|\le r|\calX|\log^5 n.
\end{equation}
We say that $\calX$ is \emph{scattered} if $\calA$ is $r$-scattered
  for every $\calA \subseteq \calX$ with $\red{0 < }|\calA| \leq \mp$ and every $r> 0$.
\end{definition}

The parameter $r$ above should be thought of as
close to zero.  Thus the rough idea is that
in a scattered query set $\calX$, for every small subset $\calA \subset \calX$,
only a small number of points in $\calX$ that lie close to the span of
$\calA$ are compatible with $\calA$.
Recall that as discussed earlier, small subsets~$\calA$ (of size at most $\mp$)
correspond to different choices of the multi-index $J\in \N^d$ in
(\ref{eq:lindeberg}).
Intuitively, our analysis can handle
points that do not lie close to the span of $\calA$ (we make this intuition
precise in Proposition \ref{propositionneed}), and as discussed above
in Lemma \ref{lem:incomp}, points that are incompatible with $\calA$ are also
good for our analysis.
Having a query set be scattered will aid us in
bounding the sum in (\ref{eq:lindeberg});
in particular, we will show that for a scattered query set,
most multi-indices $J$ are such that
$ \Pr\hspace{0.036cm}[(\bR_{-i})|_J \in \red{\calB_{J}}]
\lessapprox \epsilon^{\# J}$.
This will result in a substantially better bound in (\ref{eq:lindeberg}).

We now state the main lemma, which describes the effect of our pruning procedure:

\begin{lemma}\label{lem:pruning}
Fix $c>0$, and let $\mp=\mp(c)$ be as defined in \emph{(\ref{eq:mp})}.
Given a query set $\calX\subseteq \{\pm 1/\sqrt{n}\}^n$ with $|\calX| \leq n^{1/2-c}$,
there exists a scattered query set $\calX^*\subseteq \calX$ \emph{(}so $|\calX^*|\le |\calX|$\emph{)} such that
$$\duo(\bS,\bT)\le \duo(\bS^*,\bT^*)+0.01,$$
where $\bS^*=\sum_j \bu_j \calX^{*(j)}$ and $\bT^*=\sum_j \bv_j \calX^{*(j)}$.
\end{lemma}
Assuming Lemma \ref{lem:pruning},
  it now suffices to show that $\duo(\bS^*,\bT^*)\le 0.09$  for any
  scattered query set $\calX^*\subseteq \{\pm 1/\sqrt{n}\}^n$ of size $|\calX^*| \leq n^{1/2-c}$,
  which we will do in the following sections.

The basic step of our pruning procedure is quite
straightforward:

\begin{framed}
\begin{flushleft}\begin{enumerate}
\item[] $\!\!\!\!\!\!\!\!\!\textsf{Pruning}(\calX)$:
\item If $\calX$ is not scattered, find any
pair $(\calA,r)$ with $\calA \subseteq \calX,$ $\red{0<}|\calA| \leq \mp$, $r>0$ such that
$\calA$\\ is not $r$-scattered (i.e. (\ref{eq:scattered}) is violated).
For any such $\calA$ choose the largest possible\\ $r$
which violates (\ref{eq:scattered}).\vspace{-0.06cm}
\item Let $\calR=(\calX\cap B_{\ell_2}(\span(\calA),r))\setminus \calA$ and
let $\calR=\calR_{cover}\cup \calR_{remove}\cup \calR_{incomp}$ denote\\ the partition
as promised by Lemma \ref{lem:partition}.
\item Remove all points of $\calR_{remove}$ from $\calX$.
\end{enumerate}\end{flushleft}
\end{framed}

Given a query set $\calX\subset \{\pm 1/\sqrt{n}\}^n$, we can iteratively prune $\calX$ via the $\textsf{Pruning}$ procedure above until we obtain a scattered query set as defined in Definition~\ref{def:scattered}.
Starting with a query set $\calX$ with $|\calX| \leq n^{1/2-c}$,
we write $\calX=\calX_0\supset \calX_1\supset\cdots\supset \calX_t=\calX^*$
to denote the sequence~of
query sets we get from calling \textsf{Pruning} repeatedly until ${\calX^*}$ is scattered.
Note that the final set $\calX^*$ will be nonempty,
because $\calA\cap \calR=\emptyset$ \blue{for the sets $\calA, \calR$
used in the final application of
\textsf{Pruning}} and thus $\calA$ remains in $\calX$
  at the end of \textsf{Pruning}.\ignore{\rnote{I added a condition in the definition
of ``scattered'' that $A$ must be nonempty; maybe strictly we don't
need it but it makes this step clearer, at least to me.  Feel free to change
back if you don't like it. \green{Xi: It looks great.}}}

To prove Lemma \ref{lem:pruning}, we show that $\calX^*$ is almost
as effective as $\calX$ in the following sense:

\begin{claim}\label{claimmain}
$\duo(\bS,\bT)\le \duo(\bS^*,\bT^*)+0.01$.
\end{claim}
\begin{proof}
For each $i=0,1,\ldots,t-1$, let $(\calA_i,r_i)$ denote the pair
  identified in Step~1 of \textsf{Pruning}, when it is run on $\red{\calX}_i$.
From (\ref{eq:scattered}) we have
$|\red{\calX}_i|-|\red{\calX}_{i+1}|>r_i|\red{\calX}_i| \log^5 n.$
On the other hand, we have
$$
\sum_{i=0}^{t-1}\frac{|\red{\calX}_i|-|\red{\calX}_{i+1}|}{|\red{\calX}_i|}
\le \sum_{i=0}^{t-1} \left( \frac{1}{|\red{\calX}_{i+1}|+1} +\ldots+
\frac{1}{|\red{\calX}_i|}\right)
 =O\hspace{0.015cm}(\log |\red{\calX}|)=O\hspace{0.015cm}(\log n).
$$
We conclude that $\sum_i r_i=O\hspace{0.015cm}(1/\log^4 n)$.
Next, let
\begin{eqnarray*}
&\bS_i=\sum_j \bu_j\red{\calX}_i^{(j)},\ \ \ \bT_i= \sum_j \bv_j
\red{\calX}_i^{(j)},\ \ \
  \bS_{i+1}=\sum_j \bu_j \red{\calX}_{i+1}^{(j)}\ \ \ \text{and}\ \ \
  \bT_{i+1}=\sum_j \bv_j \red{\calX}_{i+1}^{(j)}.&
\end{eqnarray*}
We compare $\duo(\bS_i,\bT_i)$ and $\duo(\bS_{i+1},\bT_{i+1})$, and
  our goal is to show that
\begin{equation}\label{eq:bst}
\duo(\bS_i,\bT_i)\le \duo(\bS_{i+1},\bT_{i+1})+O\big(\red{(r_i+\eps)}\log^3 n+(1/\sqrt{n})\big).
\end{equation}
It then follows that
\begin{eqnarray*}
&\duo(\bS,\bT)\le \duo(\bS^*,\bT^*)+O(t/\sqrt{n}+\red{t\hspace{0.02cm}\epsilon}
\blue{\log^3 n})+\sum_i O(r_i\log^3 n)<\duo(\bS^*,\bT^*)+0.01,&
\end{eqnarray*}
since we have $\sum r_i=O\hspace{0.015cm}(1/\log^4 n)$, $t\leq n^{1/2-c}$
  and $\epsilon=n^{4/\mp-1/2}$ with $\mp\ge 5/c$ (as defined in (\ref{eq:mp})).

Fix an $i$ and let $\calR
=(\calX_{i}\cap B_{\ell_2}(\span(\calA_i),r_i))
\setminus \calA_i$ and write
$\calR=\calR_{cover}\cup \calR_{remove}\cup \calR_{incomp}$
as in Lemma \ref{lem:partition}.
Let $\calO_i$ be a union of orthants in $|\calX_i|$-dimensional
space with
\begin{eqnarray*}
&\duo(\bS_i,\bT_i)=\big|\hspace{-0.02cm}\Pr[\bS_i \in \calO_i]-\Pr[\bT_i\in\calO_i]
\hspace{0.02cm}\big|.
\end{eqnarray*}
Given $\calO_i$,
below we define a union of orthants $\calO_{i+1}$ in $|\calX_{i+1}|$-dimensional
space.  We will then show that $\calO_{i+1}$ satisfies
(\ref{eq:uo}) below and thereby obtain (\ref{eq:bst}).
\ignore{\rnote{Something is weird with the equation numbering here -- for
some reason it doesn't seem to be displaying a number by
the equation we want (equation ``eq:uo''),
and it is displaying the number of the following
equation instead?? \green{Xi: Fixed now :).}}}

We start with some terminology.
Recall that an orthant in $|\red{\calX}_i|$-dimensional space can be
viewed as an assignment of a $\{\pm 1\}$ value to each element of $\red{\calX}_i$.
Given $V\in \calX_i$ and an orthant $\calT$ in $|\calX_i|$-dimensional space,
we let $\calT(V)\in\{\pm 1\}$ denote the value assigned to $V$ by $\calT$.
We say an orthant $\calT$
  in $|\red{\calX}_i|$-dimensional space (but not necessarily in $\calO_i$)
  is \emph{bad} if
  there exist  $W\in \calR_{remove}, V\in \calR_{cover}$ such that
  $\|V-W\|_2\le 4\hspace{0.02cm}\red{r_i}$ but
  \blue{$\calT(V)\ne \calT(W)$};
  otherwise we say $\calT$ is a \emph{good} orthant.
\blue{Observe that by Lemma \ref{lem:partition}, a good orthant $\calT$
  is uniquely determined by its values $\calT(V)$, $V\in \calX_{i+1}$.}
We let $\calO_{\red{i,b}}$ denote the union of bad
  orthants in $\calO_i$, and let $\calO_{\red{i,g}}$
  denote the union of good orthants in $\calO_i$.  
  
\green{As we will see below in Claim \ref{claim111}, the probability
  of $\bS_i$ or $\bT_i$ lying in a bad orthant is negligible. 
Thus, most of $|\Pr[\bS_i \in \calO_i]-\Pr[\bT_i\in\calO_i]|$
  comes from good orthants of $\calO_i$.
Inspired by this, we will take $\calO_{i+1}$ to be the 
  projection of good orthants of $\calO_i$ onto the
  $|\calX_{i+1}|$-dimensional space.}

\green{We define formally $\calO_{i+1}$ as follows.}
We say orthants $\calT$ and $\calT'$ in $|\red{\calX}_i|$- and $|\red{\calX}_{i+1}|$-dimensional space,
respectively, are \emph{consistent} if every $V\in
\red{\calX}_{i+1}$ satisfies \blue{$\calT(V)=\calT'(V)$}.
Given $\calO_i$, we define $\calO_{i+1}$ to be the union
  of orthants in $|\calX_{i+1}|$-dimensional space
  each of which is consistent with a good orthant of $\red{\calO_{i}}$.
By definition, there is a bijection between
  orthants of $\calO_{i+1}$ and good orthants of $\calO_i$.
For each orthant $\calT'$ of $\calO_{i+1}$, we let $g(\calT')$ denote
  the corresponding good orthant $\calT$ of $\calO_i$;
let $b(\calT')$ denote the \emph{union} of all \blue{bad}
  \red{$|\red{\calX}_{i}|$-dimensional} orthants $\calT$ (not necessarily in $\calO_i$)
   that are consistent with
  $\calT'$. 
\vspace{0.02cm}

We delay the proof of the following claim:

\begin{claim}\label{claim111}
Let $\calO^*$ denote the union of all bad orthants in $|\red{\calX}_i|$-dimensional space. Then
\begin{equation}\label{claimeq}
\Pr\hspace{0.03cm}[\bS_i\in \calO^*],\hspace{0.04cm}\hspace{0.03cm}
\Pr\hspace{0.03cm}[\bT_i\in \calO^*]=O\hspace{0.015cm}\big(\red{(r_i+\epsilon)}\log^3 n\big)+O\hspace{0.015cm}(1/\sqrt{n}).
\end{equation}
\end{claim}

Returning to the proof of Claim \ref{claimmain},
  for each orthant $\calT'$ in $\calO_{i+1}$ we have
\begin{equation}\label{eq:needsecond}\begin{aligned}
\Pr\hspace{0.03cm}[\bS_{i+1}\in \calT']&= \Pr\hspace{0.03cm}[\bS_i\in g(\calT')]+
  \Pr\hspace{0.03cm}[\bS_{i}\in b(\calT')]\ \ \ \ \text{and}\\[0.3ex]
\Pr\hspace{0.03cm}[\bT_{i+1}\in \calT']&=\Pr\hspace{0.03cm}[\bT_i\in g(\calT')]+
  \Pr\hspace{0.03cm}[\bT_i\in b(\calT')].
\end{aligned}\end{equation}
Combining (\ref{claimeq}) and (\ref{eq:needsecond}), we have\vspace{0.06cm}
\begin{eqnarray}
\big|\hspace{-0.03cm}\Pr[\bS_i\in \calO_i]-\Pr[\bT_i\in \calO_i]\big| \hspace{-0.26cm}&\le
\big|\hspace{-0.03cm}\Pr[\bS_i\in \calO_{\red{i,b}}]\big|+
\big|\hspace{-0.03cm}\Pr[\bT_i\in \calO_{\red{i,b}}]\big|
  +\big|\hspace{-0.03cm}\Pr[\bS_i\in \calO_{\red{i,g}}]-\Pr[\bT_i\in \calO_{\red{i,g}}]\big|
\nonumber \\[0.8ex]
&\hspace{-3cm}\le O\hspace{0.02cm}\big(\red{(r_i+\eps)}\log^3 n+
 1/\sqrt{n}\hspace{0.03cm}\big)+\big|\hspace{-0.03cm}\Pr[\bS_{i+1}\in \calO_{i+1}]-\Pr[\bT_{i+1}\in \calO_{i+1}]\big| \nonumber \\[0.7ex]
&\ \ \ \ \hspace{-3cm}+\left|\hspace{0.05cm}\sum_{\text{$\calT'$ in $\calO_{i+1}$}}
\Pr[\bS_i\in b(\calT')]\hspace{0.05cm}\right|+
\left|\hspace{0.05cm}\sum_{\text{$\calT'$ in $\calO_{i+1}$}}
\Pr[\bT_i\in b(\calT')]\hspace{0.05cm}\right|\nonumber \\[0.9ex]
&\hspace{-3cm}=O\hspace{0.02cm}\big(\red{(r_i+\epsilon)}\log^3 n+
  1/\sqrt{n}\hspace{0.03cm}\big)+\big|\hspace{-0.03cm}\Pr[\bS_{i+1}\in \calO_{i+1}]-
\Pr[\bT_{i+1}\in \calO_{i+1}]\big|, 
\label{eq:uo}\end{eqnarray}
where the sums are over all orthants $\calT'$ in $\calO_{i+1}$.
The last inequality used (\ref{claimeq}) as well as the fact
  that the $b(\calT')$'s are unions of disjoint bad orthants in
$|\red{\calX}_i|$-dimensional space.

This finishes the proof of Claim \ref{claimmain}.
\end{proof}

It remains to prove Claim \ref{claim111}.

\begin{proof}[Proof of Claim \ref{claim111}]
We focus on $\Pr[\bS_i\in \calO^*]$ since the same argument works for $\Pr[\bT_i\in \calO^*]$.

By the definition of bad orthants in $|\red{\calX}_i|$-dimensional space, we have
$$
\Pr[\bS_i\in \calO^*]\le \sum_{V\in \red{\calR}_{cover}} \Pr\Big[\hspace{0.03cm}\text{$\exists\hspace{0.05cm}
  W\in \red{\calR}_{remove}$: $\|V-W\|_2\le 4\hspace{0.02cm}\red{r_i}$ but $\bS_i$ has different signs on $V,W$}\Big].
$$
Observe that the number of terms in the sum is $|\red{\calR}_{cover}|=O_{\mp}(1)$.

Fix a $V\in \red{\calR}_{cover}$.
\blue{By Lemma \ref{lem:hoeffding}, we have for every $W\in \calR_{remove}$
  such that $\|V-W\|_2\le 4\hspace{0.02cm}r_i$:}
\begin{eqnarray}\label{eq:used-once}
&\Pr\Bigg[\hspace{0.04cm}\Big|\sum_j \uu_j(V_j-W_j)-\E[\uu]\cdot \sum_j (V_j-W_j)
\hspace{0.05cm}\Big|
\ge 4\hspace{0.02cm}r_i\cdot \red{(\log n)^{3/4}} \Bigg]\le&\hspace{-0.25cm} \frac{1}{n^{\omega(1)}}.
\end{eqnarray}
Since the number of such $W$ is at most $n^{1/2-c}$, we have
\begin{equation}\label{eq:upupup}
\Pr\Big[\hspace{0.04cm}\exists\hspace{0.05cm}W\in \red{\calR}_{remove}
  \text{~s.t.~} \|V- W\|_2\le 4\hspace{0.02cm}r_i\ \text{and satisfies
  the condition in (\ref{eq:used-once})}\hspace{0.04cm}
\Big]\le \frac{1}{n^{\omega(1)}}.
\end{equation}

On the other hand,
  using the standard 1-dimensional Berry--Ess\'een Theorem (Theorem \ref{thm:be}),
and recalling that $\|V\|_2=1$ and each $\bu_j$ has variance 1, we have
\begin{eqnarray}\label{eq:condition}
&\Pr\Bigg[\left|\sum_{j} \uu_j V_j \right|\ge \red{(r_i+\epsilon)}\log^3 n\Bigg]=1-\blue{O}
  \hspace{0.02cm}\big(\red{(r_i+\epsilon)}\log^3 n\big) -\blue{O}\hspace{0.02cm}({1}/{\sqrt{n}}).&
\end{eqnarray}
By Lemma \ref{lem:partition}
   every $W\in \red{\calR}_{remove}$ with $\|V-W\|_2\le 4\hspace{0.02cm}\red{r_i}$ satisfies
$\big|\hspace{-0.04cm}\sum_{j}(V_j-\red{W}_j)\big|\le (r_i+\epsilon) \log^2 n.$
Combining this with (\ref{eq:upupup}) and (\ref{eq:condition}), we have that the probability of
\begin{eqnarray*}
&\sign\Big(\hspace{-0.05cm}\sum \uu_j V_j\Big)= \sign\Big(
\hspace{-0.05cm}\sum \uu_j W_j\Big),\ \ \ \text{for
  all $W\in \red{\calR}_{remove}$ with $\|V-W\|_2\le 4\hspace{0.02cm}r_i$,}&
\end{eqnarray*}
 is at least ${1-O(\red{(r_i+\epsilon)}\log^3 n)-O(1/\sqrt{n})}$.
Claim \ref{claim111} then follows.
\end{proof}

\section{A lower bound against scattered query sets} \label{sec:lbscattered}

\orange{With the pruning procedure of the previous section in hand --- showing how an arbitrary query set can be pruned so as to make it scattered ---  we now focus on proving a lower bound against scattered query sets via the approach outlined in Section~\ref{sec:goingbeyond}. In Section~\ref{sec:puttogether} we will see how such a lower bound along with our analysis in the previous section can be easily combined to complete the proof of our main theorem, Theorem~\ref{thm:main}.}

\orange{We briefly recall the setup of our approach.} Fix a $c>0$, and let $\mp$ and
  $\ell$  be defined as in (\ref{eq:mp}).
\blue{Recall that $\eps=n^{4/\mp-1/2}$ and $\delta=1/\sqrt{n}$.}
Let $\red{\calX=\{X^{(1)},\ldots,X^{(d)}\}}\subseteq \{\pm 1/\sqrt{n}\}^n$ denote a query set
  with $d\le n^{1/2-c}$.
Let $\bu_j$ and $\bv_j$ denote independent random variables
  with the same distribution as $\bu$ and $\bv$ given in Proposition
  \ref{prop:matchmoments} and \ref{prop:negsupport} with parameters $\ell$ and $\blue{\mu=\mu(\ell)}$.
\orange{Recall that
\[ \mathbf{R}_{-i} = \sum_{j=1}^{i-1} \bv_j \red{\calX}^{(j)} + \sum_{j=i+1}^n \bu_j \red{\calX}^{(j)}\]
as defined in (\ref{eq:R-minus-i}). Revisiting our discussion in Section~\ref{sec:goingbeyond}, recall that our goal is to upper bound the quantity on the RHS of
   \blue{(\ref{eq:lindeberg})} using \blue{(\ref{eq:useme})}; to be precise we would like to show that the probability $\Pr\hspace{0.03cm}[(\bR_{-i})|_J\in \calB_J]$ is typically small for most choices of $J\in \N^d$ with $|J| = \mp+1$, where $\calB_J$ denotes the origin-centered $(\# J)$-dimensional box
  \blue{$[-\eps-\beta\delta,\eps+\beta\delta\hspace{0.03cm}]^{\#J}$} and $\beta = O_\mp(1)$ is an absolute constant that depends only on the largest value in the support of \blue{$\bu$ and $\bv$}.  We do so in this section via the following lemma:}

\begin{lemma}\label{mainlemmathis}
If $\red{\calX}$ is a scattered query set of size \blue{$d\le n^{1/2-c}$}, then
\begin{equation}\label{eq:sum}
\sum_{|J| = \red{\mp}+1} \Pr\big[(\bR_{-i})|_J \in \red{\calB}_\orange{J} \big]=
\blue{O\Big(d^{\red{\mp}+1}\cdot
  \big(1/d+\epsilon \log^6 n\hspace{0.03cm} \big)^{\red{\mp}}\Big)},
\end{equation}
where $\red{\calB}_\orange{J}$
denotes the origin-centered $(\# J)$-dimensional box $\calB_J = \red{[-\eps-\beta\delta,
  \eps+\beta\delta\hspace{0.03cm}]^{\#J}}$.
\end{lemma}

Instead of focusing on the sum in (\ref{eq:sum}) we let
  \orange{$\bI=(\blue{\bV}^{(1)},\ldots,\blue{\bV}^{({\red{\mp}+1})})$}
  denote a sequence of $\red{\mp}+1$ points
  sampled from $\red{\calX}$ uniformly at random, with replacement.\ignore{\lnote{Commented out ``We need some notation. Let $\#I$ denote the number of distinct points in $I$
  and let $(\bR_{-i})_I$ denote the projection of $\bR_{-i}$
  onto the coordinates that correspond to points in $I$.'', since we already introduced this in the preliminaries. Totally okay with adding it back in to remind the reader though... \blue{Xi: I added
  it back since $I$ is a tuple but not a multi-index. We can also
  more this to the preliminaries.}}}
Let $\#\bI$ denote the number of distinct points in $\bI$,
  and $(\bR_{-i})_\bI$ denote the projection of $\bR_{-i}$
  onto the coordinates that correspond to points in $\bI$.
Lemma \ref{mainlemmathis} then follows directly from the following lemma, \orange{as
  the distribution of $\bI$ is close to the uniform distribution over $J$ with $|J|=\mp+1$}. \ignore{\rnote{Revisit this with more detail.}}

\begin{lemma}\label{mainlemmathat}
If $\red{\calX}$ is a scattered query set of size \blue{$d\le n^{1/2-c}$}, then
\begin{equation}\label{eq:expectation}
\Ex_\bI\Big[ \Pr\big[(\bR_{-i})|_\bI \in \blue{\red{\calB}_\bI} \big]
\hspace{0.02cm}\Big]=\blue{O\Big(\big(1/d+\epsilon \log^6 n\big)^{\red{\mp}}\Big)},
\end{equation}
\blue{where $\blue{\red{\calB}_\bI}$ denotes the origin-centered
  $(\#\bI)$-dimensional box $\blue{\red{\calB}_\bI}= [-\eps-\beta\delta,\eps+\beta\delta\hspace{0.03cm}]^{\#\bI}$.}
\end{lemma}
\begin{proof}[Proof of Lemma \ref{mainlemmathis}]
\blue{Let $g$ denote the following natural map
  from $\calX^{\mp+1}$ to $\{J\in \N^d:|J|=\mp+1\}$:
$$
g(I)=(J_1,\ldots,J_d),
$$where
  $J_i$ is the number of times $X^{(i)}$ appears in $I=(V^{(1)},\ldots,V^{(\mp+1)})\in \calX^{\mp+1}$.

It is clear that $g$ is surjective. As a result, we have
$$
\sum_{|J| = \red{\mp}+1} \Pr\big[(\bR_{-i})|_J \in \red{\calB}_\orange{J} \big]
\le \sum_{I\in \calX^{\mp+1}} \Pr\big[(\bR_{-i})|_I \in \red{\calB}_\orange{I} \big]
=d^{\mp+1}\cdot \Ex_\bI\Big[\hspace{0.01cm} \Pr\big[(\bR_{-i})|_\bI \in \blue{\red{\calB}_\bI} \big]
\hspace{0.02cm}\Big].
$$
Lemma \ref{mainlemmathis} then follows directly from Lemma \ref{mainlemmathat}.}
\end{proof}

Let $I=(V^{(1)},\ldots,V^{({\red{\mp}+1})})\in \red{\calX}^{\red{\mp}+1}$.
For each $j\in \blue{[2:{\mp}+1]}$, we let
\begin{equation}\label{defd}
d_j=d_{\ell_2}\big(V^{({j})},\span\{V^{(1)},\ldots,V^{({j-1})}\}\big),
\end{equation}
 and define $\eta_j$ for each $j\in [2:{\mp}+1]$ as\vspace{0.03cm}:
\begin{equation}\label{defeta}
\eta_j =\left\{
\begin{array}{ll}0 & \text{if $V^{(j)}$ is incompatible with $\{V^{(1)},\ldots,V^{({j-1})}\}$}\\[0.6ex]
1 & \text{if $V^{(j)}$ is compatible with $\{V^{(1)},\ldots,V^{(j-1)}\}$ but $d_j< \eps$}\\[0.5ex]
 \eps/d_j & \text{otherwise}.
\end{array}
\right.\vspace{0.06cm}
\end{equation}
We will use the following proposition to bound $\Pr\hspace{0.02cm}[(\bR_{-i})|_I \in \red{\calB_I}
  \hspace{0.02cm}]$:

\def\dd{\boldsymbol{d}} \def\eeta{\boldsymbol{\eta}}

\begin{proposition}\label{propositionneed}
Given an $I=(V^{(1)},\ldots,V^{({\red{\mp}+1)}})\in \red{\calX}^{\red{\mp}+1}$,
we have
\[  \Pr\big[(\bR_{-i})|_I \in \red{\calB_I}\big] \le
  \blue{O\left(\hspace{0.04cm} \prod_{j=2}^{\red{\mp}+1} \eta_j\right)+O\left(\frac{1}{n^{\red{\mp}}}\right)},\]
\blue{where $\calB_I$ denotes the origin-centered
  $(\#I)$-dimensional box $\calB_I= [-\eps-\beta\delta,\eps+\beta\delta\hspace{0.03cm}]^{\#I}$.}
\end{proposition}

We delay the proof of Proposition \ref{propositionneed} to the next section,
  but first use it to prove Lemma \ref{mainlemmathat}.

\def\bd{\mathbf{d}}

\begin{proof}[Proof of Lemma \ref{mainlemmathat} \orange{assuming Proposition~\ref{propositionneed}}]
Let $\bd_j$ and $\eeta_j$ denote two random variables
  defined from $\bI$ in the same fashion as (\ref{defd}) and (\ref{defeta}).
By Proposition \ref{propositionneed}, it suffices to show that
$$
\Ex_\bI \left[\prod_{j=2}^{\red{\mp}+1} \eeta_j\right]=\blue{O\Big(\big(1/d+\epsilon \log^6 n\big)^{\red{\mp}}\Big)}.
$$
Note that $\eeta_j$ is
  a nonnegative random variable with $\Pr[\eeta_j\le 1]=1$.

Fix $j\in [2:\mp+1]$. Let $(V^{(1)},\ldots,V^{(j-1)})\in \calX^{j-1}$
  denote a possible outcome of $(\bV^{(1)},\ldots,\bV^{(j-1)})$
  and let $\calA$ denote the \emph{set} that consists of
  $V^{(1)},\ldots,V^{(j-1)}$, so $|\calA|\le j-1\le \mp$.
For any $r> 0$, let
$$\red{\calR}=\big(\red{\calX}\cap B_{\ell_2}(\span(\orange{\calA}),r)\big)\setminus \orange{\calA}\ \ \ \text{and}\ \ \
\red{\calR}=\red{\calR}_{cover}\cup \red{\calR}_{remove}\cup \red{\calR}_{incomp}$$ denotes
the three-way partition of $\red{\calR}$ promised by Lemma \ref{lem:partition}.
By definition, we have
$$
\big|\hspace{0.03cm}\red{\calR}_{cover}\cup \red{\calR}_{remove}
\hspace{0.03cm}\big|\le rd\log^5 n+2^{\red{\mp}^2}=rd\log^5 n +O_\mp(1).
$$
This implies that, conditioning on $\orange{\calA}$
  being the set of the first $j-1$ points sampled in $\bI$:
\begin{equation}\label{probeq}
\Pr\Big[\hspace{0.03cm}\boldsymbol{d}_\blue{j}\le r\ \text{and $\bV^{(i)}$ is
compatible with $\orange{\calA}$}
  \mid \hspace{-0.03cm}\orange{\calA}\hspace{0.03cm}\Big]\le r\log^5 n+O_\mp(1/d),
  \ \ \ \text{for all $r>0$.}
\end{equation}
By the definition of $\eeta_j$, we have (note that the smallest nonzero
  value for $\eeta_j$ is $\epsilon/2$)
\begin{equation}\label{eq:cont}\begin{aligned}
\mathbf{E}\big[\eeta_j\hspace{-0.06cm}\mid\hspace{-0.06cm}
\orange{\calA}\big]&=\int_0^1
\Pr\big[\eeta_j\ge x\hspace{-0.06cm}\mid\hspace{-0.06cm} \orange{\calA}\big] \hspace{0.04cm}dx
  \le (\epsilon/2)+\int_{\epsilon/2}^1 \Pr\big[\eeta_j\ge x
  \hspace{-0.06cm}\mid \hspace{-0.06cm}
  \orange{\calA}\big] \hspace{0.04cm}dx
\end{aligned}\end{equation}
By the definition of $\eeta_j$ we have for any $x: \epsilon/2\le x\le 1$:
$$\begin{aligned}
\Pr\big[\eeta_j\ge x\hspace{-0.06cm}\mid\hspace{-0.06cm}
 \orange{\calA}\big]&\le 
\green{\Pr \big[{\dd_j} \le ({\epsilon}/{x})\ \text{and $\bV^{(j)}$ is compatible with $\calA$}
  \hspace{-0.06cm}\mid \hspace{-0.06cm}\orange{\calA}\big].}
\end{aligned}$$
It follows from (\ref{probeq}) that $\Pr\hspace{0.04cm}[\eeta_j\ge x
\hspace{-0.06cm}\mid\hspace{-0.06cm} \orange{\calA}]
  \le 
\green{(\epsilon/x)\log^5 n+O_\mp(1/d)}$. Continuing from (\ref{eq:cont}):
$$
\mathbf{E}\big[\eeta_j\hspace{-0.06cm}\mid\hspace{-0.06cm} \orange{\calA}\big]
\le (\epsilon/2)+\int_{\epsilon/2}^1 \Big(
(\epsilon/x)
  \log^5 n+O_\mp(1/d)\Big) \hspace{0.04cm}dx
=O_\mp\big(1/d+\epsilon \log^6 n\big),
$$
since $\epsilon=n^{4/h-1/2}$.
As a consequence, we have for any $j\in [2:h+1]$,
 $$\begin{aligned}
\mathbf{E}\hspace{0.03cm}[\eeta_1 \ldots \eeta_j] = \mathbf{E}
\Big[\eeta_1 \cdots \eeta_{j-1} \cdot \mathbf{E}\hspace{0.02cm}\big[\hspace{0.03cm}\eeta_j
\hspace{-0.06cm}\mid\hspace{-0.06cm}
\bV^{(1)} , \ldots, \bV^{(j-1)}\big]\Big]  =
O_{\mp}\big(1/d+\epsilon \log^6 n\big) \cdot   \mathbf{E}\hspace{0.03cm} [\eeta_1 \cdot \ldots \eeta_{j-1}].
 \end{aligned}$$
This finishes the proof of the lemma.
\end{proof}

\subsection{Proof of Proposition \ref{propositionneed}} \label{sec:propneed}
\def\RR{\mathbf{R}}

Recall that $\red{\bQ^{(i)}}$ 
  denotes the following random variable
  that is very close to $\RR_{-i}$:
$$
\red{\bQ^{(i)}}=\sum_{j=1}^i \vv_j \red{\calX}^{(j)}+\sum_{j=i+1}^{n} \uu_j
\red{\calX}^{(j)}.
$$
To prove Proposition \ref{propositionneed},
  it suffices to show that
\begin{equation}\label{finaleq}
\Pr\big[(\bQ^{(i)})|_I \in \red{\calB^*_I}\big] \le
  \blue{O\left(\prod_{j=2}^{\red{\mp}+1} \eta_j\right)}+\blue{O\left(\frac{1}{n^{\red{\mp}}}\right)},
\end{equation}
  where $\blue{\calB^*_I}$ denotes the origin-centered
  $(\#I)$-dimensional box $[-2\eps,2\eps]^{\#I}$.
This is because the entry-by-entry difference between
  $\blue{\bQ^{(i)}}$ and $\bR_{-i}$ is at most $\beta\delta$,
  so we just need to make the box $\blue{\calB_I^*}$ bigger
  than the original box $\blue{\calB_I}$ in Proposition \ref{propositionneed}
  (since $2\eps>\eps+2\beta\delta$).

We prove (\ref{finaleq}) in the rest of the section.
The claim is trivial if there is a $j\ge 2$ such that
  $V^{(j)}$~is incompatible with $\{V^{(1)},\ldots,V^{(j-1)}\}$:
When this happens the LHS of (\ref{eq:finaleq}) can be upper bounded by $1/n^{\omega(1)}$ using
  Lemma \ref{lem:incomp}.

Assume from now on that $V^{(j)}$ is compatible with $\{V^{(1)},\ldots,V^{(j-1)}\}$
  for all $j\in [2:\mp+1]$.
Let $\red{L}$ denote the set of $j\in [2:\mp+1]$ such that
  $\eta_j<1$. Then we have
$$
\prod_{j=2}^{\mp+1}\eta_j = \prod_{j\in L} \eta_j.
$$
When $L=\emptyset$, (\ref{finaleq}) is trivial since the product of $\eta_j$'s is $1$.
From now on, we assume that $t=|L|>1$ and~let $L=\{j_1,\ldots,j_t\}$, with $j_1<\cdots<j_t$.
For each $i\in [t]$,
  let $$\gamma_i=d_{\ell_2}\big(\blue{V^{(j_i)},\span\{V^{(1)},V^{(j_1)},\ldots,V^{(j_{i-1})}\}}\big),$$
with $\gamma_1=d_{\ell_2}(V^{(j_1)},\span\{V^{(1)}\})$ when $i=1$ (note that
  $j_1\ge 2$).
Using
$$
\gamma_i=d_{\ell_2}\big(V^{(j_i)},\span\{V^{(1)},V^{(j_1)},\ldots,V^{(j_{i-1})}\}\big)
\ge d_{\ell_2}\big(V^{(j_i)},\span\{V^{(1)},V^{(2)}\ldots,V^{(j_{i-1})}\}\big)=d_{j_i},
$$
we have
\begin{equation}\label{eq:pos}
\prod_{i=1}^t \gamma_i\ge \prod_{i=1}^t d_{j_i}
=\prod_{j\in L} \frac{\epsilon}{\eta_j}
=\epsilon^t\cdot \prod_{j=2}^{\mp+1} \frac{1}{\eta_j}>0.
\end{equation}

Let $A$ denote the $(t+1)\times n$ matrix whose row vectors
are $V^{(1)},V^{(j_1)},\ldots,V^{(j_t)}$.
Then

\begin{lemma} \label{lemma:detlb}
Matrix $A$ has full rank $t+1$, and
$$\det\big(AA^T\big) \ge \left(\prod_{i=1}^t \gamma_i\right)^2.$$
\end{lemma}
\begin{proof}
It follows directly from (\ref{eq:pos}) that $A$ has full rank.

Next we exhibit a series of transformations $U_1, \ldots, U_{t+1} \in \R^{(t+1)\times (t+1)}$ with the following  properties:  (1) $\det(U_k)=1$ for all $k\in [t+1]$; and (2)  for each $k \in [t+1]$, the off-diagonal entries of the $k \times k$ principal minor of matrix
\begin{equation} \label{eq:gsortho}
  (U_1 \cdot \ldots \cdot U_k) \cdot A \cdot A^T \cdot (U_1 \cdot \ldots \cdot U_k)
\end{equation}
  are all zero and the \red{$i$-th} diagonal entry is at least $\gamma_i^2$
  for all $i\in [k]$.
Taking $k=t+1$ and recalling that $\det(A \cdot B) = \det(A) \cdot \det(B)$, this gives
  the claim.

  To exhibit these matrices, we simply take $U_k$
  to be the lower-triangular matrix corresponding to the $k$-{th} step
  of the Gram-Schmidt orthogonalization of the first $k$ rows of $A$. This matrix has determinant~$1$, and after its action, (1) the $(k,k)$ entry of the matrix (\ref{eq:gsortho}) becomes
  at least \red{$\gamma_k^2$}  (by \red{definition of $\gamma_k$}); and (2) the off-diagonal entries of
   the $k \times k$ principal minor are 0 as claimed (by the nature of Gram-Schmidt orthogonalization).
 \end{proof}

To prove (\ref{finaleq}) it suffices to show that
\begin{equation}\label{finaleq2}  \Pr\big[\hspace{0.01cm}\blue{(\bQ^{(i)})|_{\{1\}\cup L} \in
{\calB'}}\hspace{0.03cm}\big]=
  O\left(\hspace{0.05cm} \prod_{j=2}^{\mp+1} \eta_j\right)+O\left(\frac{1}{n^{\mp}}\right),\end{equation}
where $\calB'=\blue{[-2\epsilon,2\epsilon]^{t+1}}$.
Below we let $\bQ$ denote $(\bQ^{(i)})|_{\{1\}\cup L}$ for convenience.

For this purpose, we let $\bw_1,\ldots,\bw_n$ denote independent
  Gaussian random variables $\calN(\mu(\red{\ell}),1)$ with the same first
$\red{\ell}=\mp^3$ moments as both $\bu$ and $\bv$
\red{(recall the first paragraph of Section \ref{sec:goingbeyond})}.
Let $$\bG=\sum_{j=1}^n \bw_j A^{(j)},$$
where $A^{(j)}$ denotes the $j$th column of $A$.
We prove (\ref{finaleq2}) using the following lemma:

\blue{\begin{lemma}\label{lemma:final1}
Let $\eta=\prod_{j=2}^{\mp+1}\eta_j$. Then the two random variables $\bQ^{(i)}$ and $\bG$ satisfy
$$\Pr \hspace{0.02cm}[\bG\in \red{\calB'}]\le O\hspace{0.02cm}(\epsilon\eta)
\ \ \ \text{and}\ \ \
\big|\Pr\hspace{0.03cm}[\bQ \in \red{\calB'}]-
\Pr\hspace{0.03cm} [\bG\in \red{\calB'}]\hspace{0.02cm}\big|<O\hspace{0.02cm}(\epsilon\eta)+O\big(1/n^{\mp}\big).
$$
\end{lemma}}

Proposition \ref{propositionneed} then follows
  \blue{(note that we gave out a factor of $\epsilon$ in the first term for free)}.

\begin{proof}[Proof of Lemma \ref{lemma:final1}]
For the first part, we calculate the covariance matrix of $\bG$.
Let $i_1,i_2\in [t+1]$:
$$
\Cov\big[\bw_jA^{(j)}\big]_{i_1,i_2}=\bE\big[(\bw_j A^{(j)}_{i_1}-\mu A^{(j)}_{i_1})
  (\bw_j A^{(j)}_{i_2}-\mu A^{(j)}_{i_2})\big]=A^{(j)}_{i_1}A^{(j)}_{i_2}.
$$
So we have $\Cov[\bG]_{i_1,i_2}=\sum_j A^{(j)}_{i_1}A^{(j)}_{i_2}$
  and thus, $\Cov[\bG]=AA^T$.
The first part of the lemma~then follows directly from Lemma \ref{lemma:detlb}
  and the definition of (the density of) multidimensional Gaussian distributions.

For the second part we apply Proposition \ref{bentkus-2.1}
  and the Lindeberg method over $\bQ^{(i)}$ and $\bG$, with
$$\calA=\calB'=[-2\eps,2\eps]^{t+1},\ \ \ \calA_{in}=[-2\eps+\xi,2\eps-\xi]^{t+1},\ \ \
\text{and}\ \ \ \calA_{out}=[-2\eps-\xi,2\eps+\xi]^{t+1},
$$
for some parameter $\xi:0<\xi<2\eps$ to be specified later.
We use the following two mollifiers:

\begin{proposition}
\label{product-mollifier2}
For all $\epsilon,\xi>0$ with $\xi<2\epsilon$, there exist two
  smooth functions $\Psi_{in},\Psi_{out}:\R^{t+1} \to$ $[0,1]$ with the following properties:\vspace{0.06cm}
\begin{flushleft}\begin{enumerate}
\item $\Psi_{in}(X) = 0$ for all $X\notin \calA$ and $\Psi_{in}(X)=1$ for all $X\in \calA_{in}$.\vspace{-0.06cm}

\item $\Psi_{out}(X)=0$ for all $X\notin \calA_{out}$ and $\Psi_{out}(X)=1$ for all $X\in \calA$.\vspace{-0.1cm}
\item For any multi-index $J \in \mathbb{N}^{t+1}$ such that $|J|=k$,
$\big\Vert
\Psi^{(J)}_{in} \big\Vert_{\infty},\hspace{0.06cm}
\big\Vert
\Psi^{(J)}_{out}\hspace{0.03cm} \big\Vert_{\infty} \le \ff(k) \cdot (1/\xi)^k$.\vspace{0.04cm}
\end{enumerate}\end{flushleft}
\end{proposition}
\begin{proof}
Let $\Phi_{\xi}:\bR\to [0,1]$ denote the smooth function given in Claim \ref{clm:smooth}
  (note that we replaced $\eps$ with $\xi$ in Claim \ref{clm:smooth}).
Let $\Phi_{in},\Phi_{out}:\bR\to [0,1]$ denote the following two smooth functions:
$$
\Phi_{in}(x)=\begin{cases}
\Phi_\xi(-x+2\eps)& \text{when $x\ge 0$}\\
\Phi_\xi(x+2\eps)& \text{when $x<0$}
\end{cases}\ \ \ \ \text{and}\ \ \ \
\Phi_{out}(x)=\begin{cases}
\Phi_\xi(-x+2\eps+\xi)& \text{when $x\ge 0$}\\
\Phi_\xi(x+2\eps+\xi)& \text{when $x<0$}
\end{cases}
$$
Using $\Phi_{in}$ and $\Phi_{out}$, we let $$\Psi_{in}(X)=\prod_{j\in [t+1]} \Phi_{in}(X_j)\ \ \ \text{and}\ \ \
  \Psi_{out}(X)=\prod_{j\in [t+1]} \Phi_{out}(X_j).$$
The three conditions on $\Psi_{in}$ and $\Psi_{out}$
  follow from a proof similar to that of Proposition \ref{prop:prod-molli}.
\end{proof}

Next from Proposition \ref{bentkus-2.1}, we have
\begin{eqnarray*}\begin{aligned}
 \big|\hspace{-0.04cm}\Pr[\bQ \in \calA]-\Pr[\bG\in \calA]\big| &\le \max\left\{\big|\hspace{-0.03cm}\E[\Psi_{in}(\bQ)]- \E[\Psi_{in}(\bG)]\big|,\hspace{0.04cm} \big|\hspace{-0.02cm}\E[\Psi_{out}(\bQ)]- \E[\Psi_{out}(\bG)]\big|\right\} \\[0.4ex]
 & \ \ \ \ + \max\big\{\hspace{-0.04cm}\Pr[\bG \in \calA_{out}\setminus \calA],\hspace{0.04cm} \Pr[\bG\in \calA\setminus \calA_{in}]\big\}.\end{aligned}\end{eqnarray*}
As $\bG$ is a Gaussian distribution with $\Cov[\bG]=AA^T$, we have
$$
\max\big\{\hspace{-0.04cm}\Pr[\bG \in \calA_{out}\setminus \calA],\hspace{0.04cm} \Pr[\bG\in \calA\setminus \calA_{in}]\big\}\le \Pr[\bG \in \calA_{out}]=O\left(\frac{(4\epsilon+2\green{\xi})^{\blue{t+1}}}{\epsilon^{\blue{t}}} \cdot \prod_{j=2}^{\mp+1} \eta_j\right),
$$
where we have used (\ref{eq:pos}), Lemma~\ref{lemma:detlb}, and the definition of the density of multidimensional Gaussian distributions. To bound $|\hspace{-0.03cm}\E[\Psi_{in}(\bQ)]- \E[\Psi_{in}(\bG)]|$, we apply Lindeberg's method again
  and follow the same argument as in \ref{sec:lindeberg} but this time match \emph{all
  the first $\ell$ moments}, which gives us that
$$
\big|\hspace{-0.03cm}\E[\Psi_{in}(\bQ)]- \E[\Psi_{in}(\bG)]\big|\le
n\cdot \frac{O_\ell(1)}{\xi^{\red{\ell}+1}}\cdot \frac{1}{n^{(\red{\ell}+1)/2}}.
$$
The same bound also holds for $|\hspace{-0.02cm}\E[\Psi_{out}(\bQ)]- \E[\Psi_{out}(\bG)]|$ by the same argument.

Combining all these inequalities, we have
$$
\big|\hspace{-0.04cm}\Pr[\bQ \in \calA]-\Pr[\bG\in \calA]\big|\le
  O\left(\frac{(4\epsilon+2\xi)^{\blue{t+1}}}{\epsilon^{\blue{t}}} \cdot \prod_{j=2}^{\mp+1} \eta_j\right)
  + n\cdot \frac{O_\mp(1)}{\xi^{\red{\ell}+1}}\cdot \frac{1}{n^{(\red{\ell}+1)/2}}.
$$
Setting $\xi=n^{2/\mp^2-1/2}< \eps$, we have
$$
n\cdot \frac{O_\mp(1)}{\green{\xi}^{\red{\ell}+1}}\cdot \frac{1}{n^{(\red{\ell}+1)/2}}
=O_\mp\left(\frac{1}{n^{2\mp-1}}\right)=O_\mp\left(\frac{1}{n^{\mp}}\right)\ \ \
\text{and}\ \ \ \frac{(4\epsilon+2\xi)^{\red{t+1}}}{\epsilon^\red{t}}<6^{\red{t+1}}\cdot \epsilon
  =O_\mp(\epsilon).
$$
This finishes the proof of the lemma.
\end{proof}

\section{Putting all the pieces together} \label{sec:puttogether}


Finally we put all the pieces together and prove our main theorem.

\begin{proof}[Proof of Theorem \ref{thm:main}]
Fix a $c>0$.
Recall that $h=h(c)\ge 5/c$, $\ell=h^3$ and $\epsilon=n^{4/h-1/2}$.

Let $\bu_1,\ldots,\bu_n$ and $\bv_1,\ldots,\bv_n$ denote independent
  random variables with the same distribution as $\bu$ and $\bv$ as given in
  Proposition \ref{prop:matchmoments} and \ref{prop:negsupport}, respectively,
  with parameters $\ell$ and $\mu=\mu(\ell)$.~Using Theorem \ref{thm:distance-to-monotonicity} in Appendix \ref{sec:nofar},
   a function drawn from $\calD_{no}$ is $\kappa(c)$-far from monotone, with probability $1-o_n(1)$,
   for some constant distance parameter $\kappa(c)$ that
  depends on $c$ only.
Thus, to prove that a non-adaptive algorithm for monotonicity
  testing requires $\Omega(n^{1/2-c})$ queries,
  it suffices to bound
$$
\duo(\bS,\bT)\le 0.1,\ \ \ \text{for all query sets $\calX\in \{\pm 1/\sqrt{n}\}^n$
  with $|\calX|\le n^{1/2-c}$},
$$
where $\bS=\sum_j \bu_j \calX^{(j)}$ and $\bT=\sum_j \green{\bv_j}\calX^{(j)}.$

Let $\calX$ denote a query set with size at most $n^{1/2-c}$.
It follows from Lemma \ref{lem:pruning} that there exists a scattered query set $\calX^*\subseteq \calX$
  such that $d=|\calX^*|\le |\calX|\le n^{1/2-c}$ and
$$\duo(\bS,\bT)\le \duo(\bS^*,\bT^*) + 0.01,$$
where $\bS^*=\sum \bu_j \calX^{*(j)}$ and $\bT^*=\sum \green{\bv_j}\calX^{*(j)}.$
Let \[
\bQ^{*(i)} = \sum_{j=1}^{i} \bv_j \calX^{*(j)} + \sum_{j=i+1}^n
\bu_j \calX^{*(j)},\]
Combining (\ref{eq:lindeberg}), (\ref{eq:useme}), and Lemma \ref{mainlemmathis}, we have
$$
\big|\hspace{-0.02cm}\E[\Psi_{\calO} (\bQ^{*(i-1)} ) ] - \E[\Psi_{\calO} (\bQ^{*(i)} )]
\hspace{0.02cm}\big|\le
\frac{O_\mp(1)}{n^{(\mp+1)/2}}\cdot \frac{1}{\epsilon^{\mp+1}}\cdot
\Big(d^{\mp+1}\cdot
  \big(1/d+\epsilon \log^6 n\big)^\mp\Big),
$$
and hence as in Section~\ref{sec:lindeberg}, summing over all $i\in [n]$ gives that
\[
|\E[\Psi_{\cal O}(\bS^\ast) - \Psi_{\cal O}(\bT^\ast)| = 
\frac {{O_\mp}(1)} {n^{(\mp-1)/2} \cdot \eps^{\mp + 1}}.
\]
Since $d\le n^{1/2-c}$, we have $d\hspace{0.03cm}\epsilon\ll 1/\log^6 n$.
By Proposition \ref{simplepro}
(and the $1$-dimensional Berry-Esseen
inequality (Theorem \ref{thm:be})
together with a union bound across the $d$ dimensions), we have
$$
\duo(\bS^*,\bT^*)\le O(d\hspace{0.02cm}\epsilon)+O(d/\sqrt{n})+\frac{O_\mp(1)}{n^{(\mp-1)/2}\cdot \epsilon^{\mp+1}}
  \cdot d =o\hspace{0.02cm}(1).
$$
It follows that $\duo(\bS,\bT)<0.1$.
This finishes the proof of the theorem.\end{proof}

\begin{flushleft}
\bibliography{odonnell-bib}{}
\bibliographystyle{alpha}
\end{flushleft}
\newpage
\appendix
\section{Standard mollifier construction}\label{onedim}
In this section we prove Claim~\ref{clm:smooth}. We begin with the following fact (see \cite{KNW:10} for a reference).
\begin{fact}\label{fact:smooth}
 There is a smooth function $b: \R \rightarrow [0,1]$ such that
 \begin{enumerate}
  \item [(i)] If $|x|>1$, then $b(x) = 0$.
  \item [(ii)] For all $\ell >0$, $\Vert b^{(\ell)} \Vert_{\infty} \le e \cdot 32^{\ell} \cdot \ell! \cdot
  \ell^{2\ell+2}$.
  \item [(iii)] $\int_{-\infty}^{\infty} b(x)\hspace{0.03cm} dx= 1$.
 \end{enumerate}
\end{fact}
Note that the bound on $\Vert b^{(\ell)} \Vert_{\infty}$ from (ii) above is $2^{-\ell -1} \cdot \alpha(\ell)$.
We now restate Claim~\ref{clm:smooth}.
\begin{noclaim*}\nonumber
 For all $\epsilon>0$, there is a smooth function $\Phi_{\epsilon} : \R \rightarrow [0,1]$ which satisfies:
 \begin{enumerate}
  \item [(1)] If $x<0$, then $\Phi_\eps(x)=0$.
  \item [(2)] If $x>\epsilon$, then $\Phi_{\epsilon}(x)=1$.\vspace{-0.06cm}
  \item [(3)] $\Vert \Phi^{(k)}_{\epsilon} \Vert_\infty \le \alpha(k)\big/\epsilon^k$.
 \end{enumerate}
\end{noclaim*}
\begin{proof}
 First, we define the function $b_{\eps}:  \R \rightarrow [0,2/\eps]$ as
 $$
 b_{\eps}(x) =\frac{2}{\epsilon} \cdot b \left(  \frac{2x}{\epsilon} \right).
 $$
 Observe that as a consequence, we have that
 $b_{\eps}$ is smooth; $b_\eps(x)=0$ if $|x|>\eps/2$;
 $\int_{-\infty}^{\infty} b_{\eps}(x)\hspace{0.03cm} dx =1$.

 Further, taking the $k$th derivative of $b_\eps$, we have
 $$
\frac{d^k b_{\eps}(x)}{dx^k} = \frac{2^{k+1}}{\epsilon^{k+1}}\cdot \left.\frac{d^k b(y)}{dy^k}\right|_{y=2x/\epsilon}
 $$
 As a result, we get that
 $\Vert b^{(k)}_{\epsilon} \Vert_\infty \le {\alpha(k)}\big/{\epsilon^{k+1}}.$
Let us define $g: \R \rightarrow \{0,1\}$ as
 $$
 g(x) =\begin{cases} 1 &\mbox{if } x> {\epsilon}/{2} \\
0 & \mbox{otherwise.} \end{cases}
 $$
 We define $\Phi_{\epsilon}= b_{\eps}\ast g$.  Since $b_{\eps} \in \mathcal{C}^{\infty}$ we have that
 $\Phi_{\epsilon} \in \mathcal{C}^{\infty}$. To see that conditions (1) and (2) of Claim~\ref{clm:smooth}
 hold, we note that
 $$
 \Phi_{\epsilon}(x) = \int_{-\infty}^{\infty}
 b_{\eps}(y)\cdot g(x-y) \cdot dy = \int_{-\eps/2}^{\eps/2}  b_{\eps}(y) \cdot g(x-y) \cdot dy.
 $$
 Note that if $x<0$, then $g(x-y) \not =0$ implies that $y<-\eps/2$. However, for $y<-\eps/2$, $b_{\eps}(y)=0$. This proves (1). If $x>\eps$, then for all $|y| \le \eps/2$, $g(x-y)=1$. Using the fact that $b_{\eps}(x)$ is a density, we get (2). Thus, it only remains to prove (3). For any $x \in \R$,
we have
 $$
\Phi_{\eps}^{(k)}(x) =
 \frac{d^k \Phi_{\eps}(x)}{dx^k} =  \big(b_{\eps}^{(k)} \ast g\big)(x)=
\int_{-\infty}^{\infty} b_{\eps}^{(k)}(y)\cdot g(x-y)\cdot dy = \int_{-\eps/2}^{\eps/2} b_{\eps}^{(k)}(y) \cdot g(x-y)\cdot dy.
 $$
Since $\Vert g \Vert_{\infty} =1$ and the interval length of the integration is $\eps$,
we get that
$$
\big\Vert  \Phi^{(k)}_{\epsilon} \big\Vert_\infty\le \epsilon\cdot \big\Vert  b^{(k)}_{\epsilon}
\big\Vert_\infty\le \alpha(k)\big/\eps^k.
$$
This completes the proof of the claim.
\end{proof}

\section{Distance to monotonicity for functions from $\calD_{no}$}
\label{sec:nofar}
In this section we prove the following theorem:
\begin{theorem}\label{thm:distance-to-monotonicity}
Let $\calD_{no}$ be the distribution over functions $\boldf(X)=\sign (\bv_1{X_1} + \cdots + \bv_n {X_n})$
where each $\bv_i$ is distributed according to $\bv$ given in Proposition \ref{prop:negsupport} with $\ell=h^3, h=h(c)$ as described
in Section \ref{sec:distributions}. Then with probability $1-o_n(1)$ over a draw of $\boldf$ from $\calD_{no}$, the function $\boldf$ is $\Omega_c(1)$-far from monotone.
\end{theorem}

As noted in Section \ref{sec:distributions}, this theorem can be proved using the methods of
\cite{CST14}; for the sake of completeness we give an alternate proof here.

Our proof uses the following claims; in all of them,  the hidden constants will depend on $c$. We will need the notion of
$\tau$-regular LTFs, which we define below.
\begin{definition}\label{def:regular-tau}
An LTF $f= \sign (v_1{X_1} + \cdots + v_n{X_n})$ is said to be \emph{$\tau$-regular} if  $|v_i|/\sqrt{\sum_{i=1}^n v_i^2} \le \tau$ for all $i \in [n]$. \end{definition}
Note that as defined above, the notion of regularity refers to a representation of an LTF and not the LTF \emph{per se}. For this section, we will blur this distinction and refer to an LTF being $\tau$-regular as long as it has a $\tau$-regular representation.

\begin{claim}\label{clm:const-neg}We have $
\Pr_{\bv_1, \ldots, \bv_n} [|\{i: \bv_i <0\}| = \Omega(n)] = 1-o_n(1).
$
As a consequence, we also have
$$
\Pr_{\bv_1, \ldots, \bv_n} \left[ \sum_{i=1}^n v_i^2 \cdot \mathbf{1}[v_i<0] = \Omega(n) \right]  = 1-o_n(1).
$$
\end{claim}
\begin{proof}
The first equation follows from item (1) in Proposition~\ref{prop:negsupport} and an application of Chernoff bound. The second equation is immediate from the first equation and the fact that the support of $\bv$ is bounded (and independent of $n$).
\end{proof}
\begin{claim}\label{clm:regularity}
A function $\boldf \sim \calD_{no}$ is $O(1/\sqrt{n})$-regular with probability $1-o_n(1)$.
\end{claim}
\begin{proof}
This follows from the first part of Claim~\ref{clm:const-neg} and the fact that the support of $\bv$ is bounded (and independent of $n$).
\end{proof}

Thus far we have been implicitly assuming the domain to be $\{-1,1\}^n$, but we may also consider the domain
$\R^n$. We recall the standard definition of the degree-1 Hermite coefficient of a function $f: \R^n \rightarrow \R$
given by an index $i \in [n]$:
$$
\tilde{f}(i) = \E_{\bX \sim \mathcal{N}^n(0,1)} [f(\bX) \cdot \bX_i].
$$
We recall the following fact that is proved in \cite{MORS10} (Proposition~25).
\begin{fact}\label{fac:MORS}
For an LTF $f = \sign(v_1 X_1 + \ldots + v_n X_n)$,  we have

$$
\tilde{f}(i) = \sqrt{\frac{2}{\pi}} \cdot \frac{v_i}{\sqrt{\sum_{i=1}^n v_i^2}}.
$$
\end{fact}
\noindent We also recall the following theorem from \cite{DDS13:KK} (Theorem~57).
\begin{theorem*}
If $f$ as defined above is $\tau$-regular, then
$$
\sum_{i=1}^n (\tilde{f}(i) - \hat{f}(i))^2 = O(\tau^{1/6}).
$$
\end{theorem*}
\noindent Combining the above theorem with Fact~\ref{fac:MORS} and Claim~\ref{clm:regularity}, we get that
$$
\sum_{i : v_i < 0} \hat{f}(i)^2 = \frac{2}{\pi} \cdot \frac{\sum_{i: v_i <0} v_i^2}{\sum_i v_i^2} - o_n(1).
$$
Combining with the second part of Claim~\ref{clm:const-neg}, we have
\begin{equation}\label{eq:weight}
\sum_{i : v_i < 0} \hat{f}(i)^2 = \Omega(1).
\end{equation}
\begin{proof}[Proof of Theorem~\ref{thm:distance-to-monotonicity}]
Let $g: \{-1,1\}^n \rightarrow \{-1,1\}$ be any monotone function. As is well known, for any $i \in [n]$ we have $\hat{g}(i) \geq 0$. We thus have
\begin{eqnarray*}\begin{aligned}
\Pr_{x \in \{-1,1\}^n} [f(x) \not  = g(x)] = \frac{1}{4} \cdot \mathbf{E}[(f(x) - g(x))^2]  &\ge \frac{1}{4} \cdot \left( \sum_{i=1}^n (\hat{f}(i) - \hat{g}(i))^2\right)  \\
 &\ge \sum_{i : v_i <0} (\hat{f}(i))^2 = \Omega(1).
\end{aligned}\end{eqnarray*}
Here the first inequality follows by Parseval's identity while the last one uses (\ref{eq:weight}).
\end{proof}

\section{Determinant of $B^{(\ell)}$}\label{sec:det}

Recall that $B^{(\ell)}$ is an $r\times r$ square matrix with $r=(\ell+1)/2$,
  whose $(i,j)$th entry is $(2(i+j)-3)!!$.
Here we prove by induction on odd $\ell$ (as $B^{(\ell)}$ is only
  defined over odd $\ell$) that
\begin{equation}\label{induction}
\det(B^{(\ell)})=\prod_{j\hspace{0.06cm} \text{odd},\hspace{0.06cm}j\in [\ell]} j!.
\end{equation}
The base case of $\ell=1$ is trivial.
Now assume for induction that the equation holds for $\ell-2$.
Given $B^{(\ell)}$, we perform the following sequence of linear transformations:
\begin{quote}
For each $j$ from $r$ down to $2$, subtract $\big[(2(r+j)-3)\times\text{column $(j-1)$}\big]$ from column $j$.
\end{quote}

Let $A$ denote the new $r\times r$ matrix. Note that $\det(A) = \det(B^{(\ell)})$. Then it is easy to verify that
  the last row of $A$ is all zero except the $(r,1)$th entry, which is $(2r-1)!!=\ell !!$;
  the $(i,j)$th entry of $A$, $i\in [r-1]$ and $j\in [2:r]$, is
$$
(2(i+j)-3)!!-(2(i+j-1)-3)!!\cdot (2(r+j)-3)
=(2(i+j-1)-3)!!\cdot (2i-2r),
$$
which is $(2i-2r)$ times the $(i,j-1)$th entry of $B^{(\ell-2)}$.

This implies that the upper right $[r-1]\times [2:r]$ submatrix of $A$
  is $B^{(\ell-\red{2})}$ after scaling the rows by $-(2r-2),-(2r-4),\ldots,-2$, respectively.
As a result, we have
$$
\det(B^{(\ell)})=\det(A)=(-1)^{r+1}\cdot\ell !!\cdot \prod_{i\in [r-1]} (2i-2r)\cdot \det(B^{(\ell-2)})
=\ell!\cdot \det(B^{(\ell-2)}).
$$
We obtain (\ref{induction}) after plugging in the inductive hypothesis for $B^{(\ell-2)}$.
This finishes the induction.
\end{document}